\documentclass[journal]{IEEEtran}
\usepackage{amsthm, amsmath,amsfonts}
\usepackage{array}
\usepackage[caption=false,font=normalsize,labelfont=sf,textfont=sf]{subfig}
\usepackage{textcomp}
\usepackage{stfloats}
\usepackage{url}
\usepackage{verbatim}
\usepackage{graphicx}
\usepackage{cite}
\usepackage[ruled,vlined]{algorithm2e}

\makeatletter
\newcommand{\removelatexerror}{\let\@latex@error\@gobble}
\makeatother
\usepackage{booktabs}
\usepackage{comment}
\usepackage{pgfpict2e}
\usepackage{float}
\usepackage{amsmath}
\newtheorem{theorem}{Theorem}
\newtheorem{lemma}{Lemma}
\newtheorem{reduction}{Reduction Rule}
\newtheorem{definition}{Definition}
\newtheorem{problem}{Problem}
\usepackage{multirow}
\usepackage{supertabular}
\let\mycomment\comment
\let\comment\undefined
\let\originalwidth\textwidth

\let\comment\mycomment
\let\textwidth\originalwidth

\begin{document}
	
	\title{A Dual-mode Local Search Algorithm for Solving the Minimum Dominating Set Problem}
	
	\author{Enqiang Zhu, Yu Zhang, Shengzhi Wang, Darren Strash, and Chanjuan Liu
		\thanks{This work was supported in part by the National Natural Science Foundation of China under Grants (61872101,
			62172072), in part by the Natural Science Foundation of Guangdong Province of China under Grant 2021A1515011940, in
			part by Science and Technology Projects in Guangzhou. \textit{(Enqiang Zhu and Yu Zhang contributed equally to this work.) (Corresponding author: Darren Strash; Chanjuan Liu)}}
		
		\thanks{Enqiang Zhu is with the Institute of Computing Science and Technology, Guangzhou University, Guangzhou 510006, China (e-mail: zhuenqiang@gzhu.edu.cn). 
			
			Yu Zhang is with the Cyberspace Institute of Advanced Technology, Guangzhou University, Guangzhou 510006, China (e-mail:
			zhangyu@e.gzhu.edu.cn).
			
			Shengzhi Wang is with the Institute of Computing Science and Technology, Guangzhou University, Guangzhou 510006, China (e-mail: wallace\_sz@163.com). 
			
			Darren Strash is with the Department of Computer Science, Hamilton College, Clinton, NY 13323, USA (e-mail: dstrash@hamilton.edu).

            Chanjuan Liu is with the School of Computer Science and Technology, Dalian University of Technology, Dalian 116024, China (e-mail: chanjuanliu@dlut.edu.cn).
            }
	}
	
	
	
	\maketitle
	
	\begin{abstract}
		Given  a graph, the minimum dominating set (MinDS) problem is to identify a smallest set $D$ of vertices such that every vertex not in $D$ is adjacent to at least one vertex in $D$.  The MinDS problem is a classic  $\mathcal{NP}$-hard problem and has been extensively studied because of its many disparate applications in network analysis. To solve this problem efficiently, many heuristic approaches have been proposed to obtain a good solution within an acceptable time limit.  However, existing MinDS heuristic algorithms are always limited by various tie-breaking cases when selecting vertices,  which  slows down the  effectiveness of the algorithms.  In this paper, we design an  efficient local search algorithm for the MinDS problem, named  DmDS---a  dual-mode local search framework that probabilistically chooses between two distinct vertex-swapping schemes. We further address limitations of other algorithms by introducing vertex selection criterion based on the frequency of vertices added to solutions to address tie-breaking cases, and a new strategy to improve the quality of the initial solution via a greedy-based strategy integrated with perturbation.  We evaluate DmDS against the state-of-the-art algorithms on seven datasets, consisting of 346 instances (or families) with up to tens of millions of vertices. Experimental results show that DmDS obtains the best performance in accuracy for almost all instances and finds much better solutions than state-of-the-art MinDS algorithms on a broad range of large real-world graphs.
		
	\end{abstract}
	
	
	\section{Introduction}	\label{sec1}
	\IEEEPARstart{A}{} dominating set (DS)---a set of vertices $D$ in which each vertex of the graph is in $D$ or adjacent to at least one vertex in $D$---is an important structure in graph theory. DSs have applications spanning the sciences, including social networks~\cite{wang2009positive}, epidemic control~\cite{zhao2020minimum}, and biological network analysis~\cite{NACHER201657, wuchty2014controllability}. One key example problem is to identify a smallest group of proteins in protein-protein interaction networks that enables each protein outside the group to be reached by interaction with a protein within the group. This problem is equal to finding a minimum dominating set (MinDS)~\cite{wuchty2014controllability}.
	
	Exploring  effective algorithms to find a MinDS  has been attracting considerable interest.  Because the standard brute-force search  algorithm for MinDS  takes $O^*(2^n)$ time, a lot of research has been carried out to design exact MinDS algorithms with improved time complexity.  
	The first such  algorithm was introduced by Fomin et al.~\cite{fomin2004exact}, and since then, many improved algorithms have been proposed~\cite{grandoni2006note, van2006design, schiermeyer2008efficiency, van2008design, van2009inclusion, van2011exact}. 
	The current state-of-the-art exact MinDS algorithm is based on the potential method, which takes $O(1.4864^n)$ time and polynomial space~\cite{iwata2012faster}.  On the other hand, due to its \textbf{NP}-hardness~\cite{garey1979computers},  researchers  have shown an increased interest in designing inexact algorithms for finding a close-to-minimum DS.  Greedy-based approaches provide a simple and efficient way for constructing DSs, but such algorithms often perform poorly (especially on large instances) and cannot satisfy practical requirements \cite{cai2021two}. To improve the quality of solutions obtained by greedy-based algorithms, a general approach is to first decrease the size of the input instance by applying exact data reduction rules, then use greedy-based approaches to generate an initial solution, and finally adopt heuristics to improve the solution quality~\cite{lamm2016finding, dahlum2016accelerating, akiba2016branch}.

	\subsection{Related Work} 
	Many recent studies show the significance of heuristic approaches to solve intractable problems, such as the partition coloring~\cite{9241416},  critical node~\cite{LIU2023119140}, and minimum vertex cover~\cite{quan2021local} problems.
	Here we give a brief review of the heuristic algorithms for solving the MinDS problem. 
	
	Heuristics, including ant-colony optimization~\cite{potluri2011two},  genetic algorithms~\cite{hedar2010hybrid, alharbi2017genetic}, and local search algorithms~\cite{hedar2012simulated, fan2019efficient, cai2021two} have been proposed to deal with the MinDS problem. In particular, local search approaches have attracted much attention in recent years. Vertex swapping is an important component in local search for graph problems. If vertex swapping depends heavily on greedy methods, then the algorithms will repeatedly visit candidate solutions that have been encountered recently and result in a local optimum. This phenomenon is called the cycling problem, which is an issue inherent in local search algorithms.   
	To address this problem, Cai et al. \cite{cai2011local} proposed a configuration checking (CC) strategy, adopting Boolean array markers to prevent recently visited vertices from being revisited,  where the restriction of a vertex $v$ is relaxed if  the state  of a vertex adjacent to  $v$ is changed. 
	In 2017, an algorithm named CC$^2$FS based on a variant of CC (CC$^2$) was proposed to solve the minimum weight dominating set (MinWDS) problem \cite{wang2017local}, in which CC was extended to  \emph{two-level} mode: a vertex $v$  is reset to an allowed state if a vertex within distance two from $v$ is removed from (or added to) the solution. Subsequently, Wang et al.~\cite{wang2018fast}  introduced another variant CC$^2$V3 of CC that adds a new state to CC$^2$ for the purpose of further distinguishing the priority of visitable vertices and showed that CC$^2$V3 can enhance the performance of  CC$^2$FS on large sparse real-world graphs. Based on CC$^2$V3,  an algorithm called FastMWDS was proposed to identify a MinWDS.
	In 2019,  by  combining score checking (SC) with probabilistic walk, Fan et al.~\cite{fan2019efficient} proposed a local search algorithm named ScBppw for the MinDS problem on large graphs, which outperforms FastMWDS on many large graphs.  Recently, Cai et al. \cite{cai2021two} introduced a two-goal  local search framework for the MinDS problem and developed a MinDS algorithm named FastDS based on three graph reduction rules (called inference rules in their article).  In contrast to the classic local search MinDS framework,  FastDS aims to improve a given $k$-sized feasible solution to a  ($k-1$)- or ($k-2$)-sized feasible solution each time, which accelerates the convergence speed and expands the search region.  To date, the state-of-the-art algorithms for MinDS are ScBppw and FastDS.

	\subsection{Motivation and Our Contributions}
	\subsubsection{Motivation}
	Note that the ScBppw algorithm can obtain small solutions on large graphs, but  it does not perform as well as FastDS  on  medium size graphs.  One reason is that the SC strategy consumes  a lot of time in the process of implementing local search. In addition,  the SC strategy is stronger than the  CC$^2$ strategy \cite{fan2019efficient}, which results in more vertices being prohibited from  visiting and limits the algorithm's search space.
	Regarding FastDS, we observe that there are many vertices with equal  priority when selecting vertices, even though the algorithm uses two criteria to address  the ties.  Moreover, the initial solution of FastDS is constructed by a greedy-based strategy associated with non-dominated degrees,  which has a great potential to be improved. Therefore, we would like to explore techniques to  improve  the above issues.
	
	\subsubsection{Contributions}
	This paper focuses on designing an efficient heuristic algorithm to find a high-quality (small) dominating set in a given graph. We propose a new local search algorithm for the MinDS problem named DmDS, which includes three main ideas: a novel local search framework, a new approach for constructing an initial solution, and a new criterion for selecting vertices. 
	First, to overcome the cycling phenomenon,  we propose a dual-mode local search framework in the search phase, which allows the algorithm to run  in two modes. That is,  when finding a $k$-sized dominating set, DmDS searches for a solution with size ($k$-1) by (2, 1)-swaps or (3, 2)-swaps, where a ($j$, $j-1$)-swap removes $j$ vertices from the solution and replaces them with $j-1$ vertices~\cite{andrade2012fast}. Second, a high-quality initial solution can accelerate the convergence speed of local search algorithms. The common methods for  generating an initial DS are greedy-based strategies \cite{cai2021two} and random-based strategies \cite{abed2022hybrid}, because these strategies are simple and efficient.  However,  the quality of solutions generated by these strategies is often poor and the time complexity would be increased significantly when adding complicated strategies to improve solutions' quality. 
	Based on the above considerations,  DmDS uses  two approaches to generate solutions and chooses  the better one as the initial solution,  including  a greedy-based algorithm according to dominating profits  and an approach integrating the  greedy strategy with a perturbation mechanism. 
	Third,  in the process of swapping vertices, selecting vertices according to a single criterion would suffer from many tie-breaking cases, i.e., there are multiple vertices with the best value.  Previous work used ``age'' to address this problem.  However, on large sparse real-world graphs, the improvement is not obvious and  there are still multiple best vertices that can be selected at the same time. Therefore, we propose a new vertex selection criterion based on the number of times each vertex has been added to the solution.
	
	\subsection{Organization of This Paper}
	The remainder of this paper is organized as follows. Section \ref{sec2} provides preliminary information, including important notation and terminology, as well as a brief review of the local search and related techniques. 
	Section \ref{sec4} describes the DmDS algorithm. Section \ref{sec5} is devoted to the design and analysis of experiments, and Section \ref{sec6} provides concluding remarks.

	\section{Preliminaries}	\label{sec2}
	We first introduce some necessary notations and terminologies,  followed by a brief review of local search. 
 
	\subsection{Notations and Terminologies}
	All graphs considered in this paper are unweighted, undirected, and simple.  We use  $G=(V, E)$ to denote a graph with \emph{vertex set} $V$  and  \emph{edge set}  $E$.  A vertex is \emph{adjacent} to (or a \emph{neighbor} of) another vertex in $G$ if they are connected by an edge. For a vertex $v\in V$, the \emph{degree} of $v$, denoted by $d(v)$,  is the number of edges incident with $v$;  the \emph{neighborhood} of $v$, denoted by $N(v)$, is the set of neighbors of $v$, i.e., $N(v)$=$\{u\mid u\in V$ and $uv\in E\}$;  the \emph{closed neighborhood} of $v$, denoted by $N[v]$, is  the union of   $N(v)$ and $\{v\}$, i.e., $N[v] = N(v) \cup \{v\}$; and the \emph{two-level closed neighborhood} of $v$, denoted by  $N^2[v]$, is the set of vertices that have distance at most two from $v$, i.e., $N^2[v] = N[v] \cup (\bigcup _{u\in N(v)}N(u))$. Moreover, for a subset $S \subseteq V$, let $N(S)= \bigcup _{v\in S} (N(v) \setminus S)$ and $N[S]=N(S)\cup S$.	
	
	\begin{definition}
		In a graph $G = (V, E)$, a vertex $u$ is dominated by a vertex $v$ if $v\in N[u]$. A dominating set of $G$ is a subset $D \subseteq V$ such that every vertex in $V$ is dominated by a vertex in $D$, i.e.,  $N[D]=V$.  A dominating set of minimum cardinality is called a minimum dominating set (MinDS). The MinDS problem is to find a MinDS in a graph.  
	\end{definition}
	
	Take the graph shown in Figure \ref{fig-example} as an example to illustrate the definition of MinDS. Observe that this graph contains six vertices and the maximum degree of vertices is four, which implies that a MinDS contains at least two vertices. On the other hand,  $N[\{v_1, v_4\}]=\{v_1, v_2, v_3, v_4, v_5, v_6\}$. Therefore,  $\{v_1, v_4\}$ is a MinDS.  Note that there are other MinDSs, such as $\{v_2,v_5\}$,   $\{v_2,v_6\}$,  and $\{v_2,v_4\}$. Thus, a graph may have more than one MinDS.

	\begin{figure}[H]
		\centering
		\includegraphics[width=0.5\linewidth]{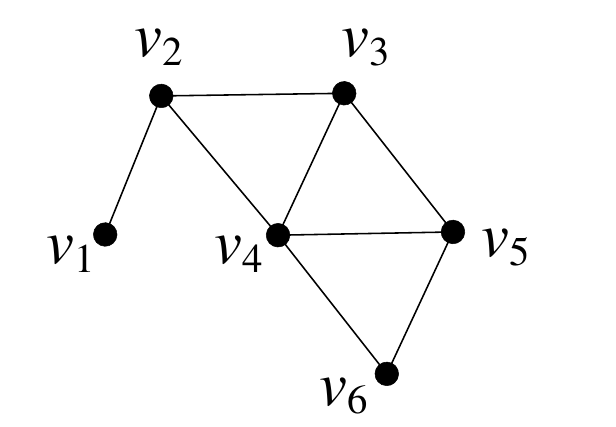}
		\caption{An example graph where $\{v_1, v_4\}$ is a MinDS.}\label{fig-example}
	\end{figure}

	\subsection{Local Search}
	The principle of a local search algorithm can be described as follows:  Start with an initial solution and then  iteratively improve it (from one solution to another)  by removing vertices, adding vertices, and swapping vertices. Local search  often incorporates some effective techniques to enhance its search results, such as tabu strategies, reduction techniques, and heuristic vertex selection strategies. For example,  CC~\cite{cai2011local}, CC$^2$~\cite{wang2017local}, CC$^2$V3~\cite{wang2018fast}, and SC~\cite{fan2019efficient} are tabu strategies that are designed to avoid wasting time in low-quality search space and  address the cycling problem by forbidding some unrelated vertices to join to the solution.  
    The aim of designing reduction rules is to  simplify instances,  with the advantage of  reducing the search space by determining a part of vertices that definitely belong (or not belong) to an optimal solution.  
	Additionally, in each iteration of the local search, the greedy-based strategies for selecting vertices are time-consuming and easily result in a local optimum. For this,  Cai et al. \cite{cai2015balance}  proposed a heuristic vertex selection strategy named Best from Multiple Selections (BMS), which selects the best vertex from  $t$ random vertices ($t$ is a predefined parameter, in general setting to 50) instead of traversing all vertices. The BMS strategy has been used to improve the efficiency of local search algorithms for solving the MinDS problem \cite{wang2018fast, cai2021two}.

	\section{Main Algorithms}	\label{sec4}
	We first present the top-level local search framework (Algorithm \ref{alg1}), followed by the construction of initial solutions (Algorithm \ref{alg2}).  Finally, we give the local search algorithm DmDS (Algorithm \ref{alg3})  with the time complexity analysis. 
	
	For presentation purposes, we specify some notations used in this section.  
	Given a graph $G$, let $D\subseteq V$ and $X\subseteq V$ be two disjoint sets, representing a partial dominating set of $G$ and a set of vertices forbidden from being added to $D$, respectively. We use a  Boolean function $dom_D$ (from $V$ to \{1,0\}) to characterize whether a vertex is dominated  by a vertex in $D$, which is defined as follows: for every vertex $v\in V$, if $v\in N[D]$, then $dom_D(v)=1$; otherwise, $dom_D(v)=0$. Denote by  $UD$ the set of vertices with $dom_D$ value 0, i.e., $UD=\{v\mid v\in V$ and $dom_D(v)=0\}$.  
	For a vertex $u \in D$, denote by $loss_{D}(u)$ the number of vertices $v$ such that $dom_D(v)=1$ and  $dom_{D\setminus \{u\}}(v)=0$; and for a vertex $w \notin D$, denote by $gain_{D}(w)$ the number of vertices $v$ such that $dom_D(v)=0$ and  $dom_{D\cup \{w\}}(v)=1$. Moreover, we use the $age(v)$ to denote the number of iterations the status of  $v$ is not changed, and $freq(v)$ to  denote the number of times  $v$ is added to the solution. The large $age$ of a vertex implies that the vertex has been in (or not in) a solution for a long time, and in this case, the vertex should have a priority to be removed from (or added to) a solution. In addition, if a vertex has a smaller $freq$ value, then it has been added to a solution fewer times, and it should have a priority to be added to a solution.

	\subsection{A Dual-mode Local Search Framework} 
	The details of the local search framework are shown in Algorithm \ref{alg1}. First, an initial solution $D$ is constructed by a procedure \emph{Initialization()} (line 1), where some reduction techniques are adopted to reduce the search space. The next is the main part of the framework  (lines 2--12), which tries iteratively reducing the size of a feasible solution  within a given period of time. 
	Specifically, each time a feasible solution $D$ is obtained, the algorithm checks  and removes all redundant vertices from $D$ (lines 3--5),  updates the  best solution $D^*$ (line 6), and removes a vertex with the minimum $loss$ value from $D$ (line 7), where a vertex $v \in D$ is redundant if $loss_D(v)=0$ and $D\setminus \{v\}$ is a feasible solution. On the other hand, when the current solution $D$ is not feasible, two vertex-swapping modes (1, 1)-swap and (2, 2)-swap are implemented with probabilities $\alpha$ (lines 9--10) and   $1-\alpha$ (lines 11--12), respectively,   to exchange vertices  until $D$ becomes feasible, where $ (k, k)$-swap is to remove $k$ vertices from $D$ and add $k$ new vertices to $D$ for $k=1,2$. Finally, the best solution $D^*$ is returned (Line 13).

	\begin{algorithm}[!t]
		\caption{A dual-mode local search framework} \label{alg1}
		\KwIn{A graph $G=(V,E)$, the \emph{cutoff} time }
		\KwOut{A dominating set $D^*$ of $G$}
		
		\nlset{1} $D \leftarrow$  \emph{Initialization}$(G)$;
		
		\nlset{2} \While{elapsed\_time $<$ cutoff}{
			\nlset{3}	\If{$D$ is feasible}
			{\nlset{4} \If{$D$ contains redundant vertices}
				{\nlset{5} remove all redundant vertices from $D$;}
				
				\nlset{6}	$D^* \leftarrow D$;
				
				\nlset{7}	remove a vertex with minimum $loss$ from $D$;
				
				\nlset{8}	continue;
			}
			
			\nlset{9}	\If{with probability $\alpha$}{\nlset{10} search by \textbf{(1, 1)-swap};}
			
			\nlset{11}	\Else{\nlset{12} search by \textbf{(2, 2)-swap};}
		}
		
		\nlset{13} \Return{$D^*$};
	\end{algorithm}

	\subsection{Initial Solution Generation}
	The graph reduction strategy is a key mechanism for designing effective heuristic algorithms for intractable graph problems on large sparse graphs \cite{lamm2016finding, zhu2022adaptive}. To generate a high-quality initial solution for DmDS,  we first  use three simple graph reduction rules, proposed by Cai et al.~\cite{cai2021two}, to reduce the  search space  (Algorithm \ref{alg2}; lines 1--4). The sets $D$ and $X$ in the following three reduction rules are the same as described above.
	
		\begin{reduction}[Degree 0]
			Let $v$ be a vertex with degree zero. If $dom_{D}(v)=0$, then add $v$ to $D$.
		\end{reduction}
		
		\begin{reduction}[Degree 1]
			Let $v$ be a  vertex of degree one and $u$ the neighbor of $v$. If $dom_{D}(u)=0$, 
			then add $u$ to $D$ and $v$ to $X$.
		\end{reduction}
		
		\begin{reduction}[Triangle]
			Let ${u,v,w}$ be a triangle such that $d(u)=d(v)=2$. If $dom_{D}(w)=0$, then add $w$ to $D$ and $u$ and $v$ to $X$.
		\end{reduction}

	As shown in Algorithm~\ref{alg2}, the sets $D^*$ and $X$ (of vertices in an optimal solution and excluded from that optimal solution, respectively) can be obtained by applying Reduction Rules 1--3. We invite the interested reader to refer to~\cite{cai2021two} for the proofs of correctness.  Next, after initializing  three sets $V, D, D'$ (line 5), two distinct strategies Greed($G,D$) and Perturbation($G,D'$) are used to construct two initial solutions $D$ and $D'$, respectively (lines 6--7), and the best one is selected as the final initial solution  (lines 8--9).  The Greed($G,D$) iteratively adds a vertex $v$ with the maximum $gain_D(v)$ to $D$ until $D$ becomes a feasible solution  ({\bf{Procedure}} Greed($G,D$); lines 1--3), and finally remove redundant vertices from $D$ ({\bf{Procedure}} Greed($G,D$); line 4). The Perturbation($G,D'$)  first selects a vertex $v$ with the maximum $gain_{D'}(v)$, saves the value of  $gain_{D'}(v)$, and updates $D'$ by adding $v$ to $D'$  ({\bf{Procedure}} Perturbation($G,D'$); lines 2--3); next, 
	if there is a vertex $u$ with minimum loss in $D'$ such that $loss_{D'}(u) < gain_{D'}(v)$, then $v$ is removed from $D'$  ({\bf{Procedure}} Perturbation($G,D'$); lines 4--6); the algorithm repeatedly conducts the above procedure until $D'$ becomes a feasible solution; finally, the algorithm removes all redundant vertices from $D'$({\bf{Procedure}} Perturbation($G,D'$); line 7).

	\begin{algorithm}[htbp] \label{alg2}
		\caption{$Initialization(G)$}
		\KwIn{A graph $G=(V,E)$ }
		\KwOut{A dominating set $D$ of $G$}
		\nlset{1} \ForEach{$u$ in $V$}{
			\nlset{2}	$D^*, X \leftarrow$ apply \textbf{Reduction Rule 1} on $u$;
			
			\nlset{3}	$D^*, X \leftarrow$ apply \textbf{Reduction Rule 2} on $u$;
			
			\nlset{4}	$D^*, X \leftarrow$ apply \textbf{Reduction Rule 3} on $u$;
		}
		
		\nlset{5} $V \leftarrow V \setminus X, D \leftarrow D^*, D' \leftarrow D^*$;
		
		\nlset{6} $D \leftarrow$ Greed$(G,D)$;
		
		\nlset{7} $D' \leftarrow$ Perturbation$(G,D')$;
		
		\nlset{8} \If{$|D'| < |D|$}{\nlset{9} $D \leftarrow D'$;}
		
		\nlset{10} \Return{{$D$};}
		
		\vspace{5pt}
		
		\textbf{Procedure} Greed($G, D$);
		
		\nlset{1} \While{there exist vertices with $dom$ is $false$}{
			\nlset{2}	$v \leftarrow $ a vertex with maximum $gain_{D}$ in $V \setminus D$;
			
			\nlset{3}	$D \leftarrow D \cup \{v\}$;
		}
		\nlset{4}   remove all redundant vertices from $D$;
		
		\nlset{5}   \Return{{$D$};}
		
		\vspace{5pt}
		
		\textbf{Procedure} Perturbation($G, D'$);
		
		\nlset{1} \While{there exist  vertices with $dom$ is $false$}{
			\nlset{2}	$v \leftarrow $ a vertex with maximum $gain_{D'}$ in $V \setminus D'$;
			
			\nlset{3}	save $gain(v)$, $D' \leftarrow D' \cup \{v\}$;
			
			\nlset{4} $u \leftarrow $ a vertex with minimum $loss_{D'}$ in $D' \setminus D^*$;
			
			\nlset{5} \If{$loss_{D'}(u) < gain_{D'}(v)$}{
				\nlset{6} $D' \leftarrow D' \setminus \{u\}$;
			}
		}
		\nlset{7}   remove all redundant vertices from $D'$;
		
		\nlset{8} \Return{{$D'$};}
	\end{algorithm}

	\subsection{The DmDS algorithm}
	The pseudocode of DmDS  is presented in Algorithm \ref{alg3}.  The algorithm starts with an initial solution $D$ obtained by Algorithm \ref{alg2} (line 1), and then improves $D$ iteratively (lines 2--17).  In each iteration,  if $D$ is a feasible solution,  the algorithm first removes the redundant vertices from $D$ (lines 3--6), and then removes a vertex  $v$ with the minimum $loss_D(v)$ from $D$  (line 7).  Next, the algorithm conducts  a tiny perturbation by randomly removing a  vertex $u_1$ from $D$  (lines 8--9), and removes another vertex $u_2$ according to the BMS heuristic with a probability $\alpha$  (lines 10--12). Note that if $u_2$ is removed, then a (2, 2)-swap is implemented; otherwise, a (1, 1)-swap is implemented.  After removing vertices from $D$, it requires to choose proper vertices from $N[UD]$  and add them to $D$ for the purpose of obtaining a  feasible solution $D$. For this, the algorithm first selects  a vertex $w_1 \in (N[UD])$ with the greatest $gain_D(w_1)$ and add it to $D$ (lines 13--14), and if there are three removed vertices in this iteration, the algorithm further adds a vertex $w_2$ (selected by the same way as $w_1$) to $D$ (lines 16--17).
	Finally, the best found solution $D^*$ is returned when the cutoff is reached. Note that there are many tie-breaking cases when selecting a vertex to be removed from $D$ or to be added to $D$.  The algorithm breaks the ties by $age$ and $freq$. Specifically, when removing a vertex from $D$, the algorithm selects vertices with large $age$ (and large $freq$ if more than one vertex has the same largest $age$); when adding a vertex to $D$, the algorithm selects vertices with large $age$ (and small $freq$ if more than one vertex has the same largest $age$).

	\begin{algorithm}[htbp] \label{alg3}
		\caption{DmDS}
		\KwIn{A graph $G=(V,E)$, the \emph{cutoff} time }
		\KwOut{A dominating set $D^*$ of $G$}
		
		\nlset{1} $D \leftarrow  Initialization(G)$;
		
		\nlset{2} \While{elapsed\_time $<$ cutoff}{
			\nlset{3}	\If{$D$ is feasible}
			{\nlset{4} \While{there exists a vertex $v$ with $loss_{D}(v)=0$}
				{\nlset{5} remove $v$ from $D$;}
				
				\nlset{6}	$D^* \leftarrow D$;
				
				\nlset{7}	remove a vertex with minimum $loss_{D}$ from $D$, breaking ties by $age$ and $freq$;
			}
			
			\nlset{8} $u_1 \leftarrow $ a random vertex in $D$;
			
			\nlset{9} $D \leftarrow D \setminus \{u_1\}$, $rm\_num \leftarrow 2$\; \tcp{(3, 2)-swap, else (2,1)-swap}
			
			\nlset{10}	\If{with probability {$\alpha$}}{	
				
				\nlset{11}	$u_2 \leftarrow $ a vertex in $D$ selected by BMS heuristic, breaking ties by $age$ and $freq$;
				
				\nlset{12}	$D \leftarrow D \setminus \{u_2\}$, $rm\_num \leftarrow 3$;
			}
			\nlset{13} $w_1 \leftarrow $ a vertex in $N[UD]$ with the greatest $gain_{D}$, breaking ties by $age$ and $freq$;
			
			\nlset{14} $D \leftarrow D \cup \{w_1\}$;
			
			\nlset{15} \If{$UD \neq \emptyset$ and $rm\_num = 3$}{
				\nlset{16}	$w_2 \leftarrow $ a vertex in $N[UD]$ with the greatest $gain_{D}$, breaking ties by $age$ and $freq$;
				
				\nlset{17}	$D \leftarrow D \cup \{w_2\}$;
			}
			
		}
		
		\nlset{18} \Return{$D^*$};
	\end{algorithm}
	
	\subsection{Complexity Analysis}
	In this section, we analyze the time complexity of DmDS. For a given graph $G=(V,E)$, let $|V|=n$ and $|E|=m$, we use an adjacency list to represent $G$, and the time complexity for constructing it is $O(n+m)$.  
	
	\begin{lemma}\label{th1}
		Reduction Rules 1--3 run in $O(n)$ time.
	\end{lemma}
	
	\begin{proof}
		The procedure traverses all vertices once to verify the conditions of the rules. The time complexity for determining whether a vertex $v$ satisfies the conditions of Reduction Rules 1 and 2 is $O(1)$. For Reduction Rule 3, computing the common neighbors of two neighboring vertices $u$ and $v$ each with degree 2 takes $O(1)$ time. Furthermore, the operation of adding a vertex to a set takes $O(1)$ time. Therefore, the overall time complexity of the procedure is $O(n)$.

	\end{proof}
	
	\begin{lemma} \label{heapTheorem}
		Algorithm \ref{alg2} runs in $O(k\log k)$, where $k = |UD|$.
	\end{lemma}
	
	\begin{proof}
		We utilize two heaps for the \emph{gain} and \emph{loss}, respectively,  to accelerate the initial solution generation, which can be constructed in  $O(k)$ time.  Note that selecting vertices with the maximum $gain$ and $loss$ values takes $O(1)$ time, and adjusting the heaps takes $O(k\log k)$ time. In addition,  Reduction Rules 1--3 run in $O(n)$ time by Lemma \ref{th1}, and deleting redundant vertices can be done in  $O(k)$ time. 
		Therefore,  the time complexity of Algorithm \ref{alg2} is $O(k\log k)$.
	\end{proof}
	
	\begin{theorem} \label{dddd}
		DmDS (Algorithm \ref{alg3}) runs in  $O(\max\{n+m, k\log k\})$,  where $k = |UD|$.
	\end{theorem}
	\begin{proof}
		First, by Lemma \ref{heapTheorem}, initial solution generation takes $O(k\log k)$ time.  Second, removing  redundant vertices takes $O(n)$. Third, removing a vertex with a minimum $loss$ from a current solution $D$ requires to traverse $D$, taking $O(n)$ time; removing a vertex randomly from $D$ and implementing the BMS heuristic are $O(1)$. 
        In the worst case, removing a vertex can result in $\Delta+1$ vertices that are not dominated,
        where $\Delta$ is the maximum  degree of $G$. 
        Note that each time a vertex is added to $D$, it has to examine $k$ vertices, implying the adding operation takes $O(k)$ time.  
        In addition, in one round of iteration, at most three vertices are removed from $D$ and at most two vertices are added to $D$. Therefore, the time complexity of each iteration is $O(n)$, and DmDS runs in $O($max$\{n+m, k\log k\})$, where $O(n+m)$ is the time complexity for constructing an adjacency list.
  
	\end{proof}

	\section{Experiments}	\label{sec5}
	We conduct experiments on an extensive collection of benchmark instances to compare the performance of different algorithms. We first introduce the seven benchmark datasets that we use, followed by a description of the experimental setup. Then, we report the experimental findings, including the size of the dominating sets obtained by  the algorithms under identical conditions, along with an assessment of their convergence performance. Finally, an empirical evaluation of the parameters in the DmDS algorithm is provided, followed by a detailed discussion of the experimental results.
	
	\subsection{Testing Benchmarks}
	DmDS is evaluated on seven benchmark datasets, which include four standard benchmarks and three benchmarks of large sparse real-world graphs. All instances are processed as undirected unweighted simple graphs.

	\emph{UDG\footnote{https://github.com/yiyuanwang1988/MDSClassicalBenchmarks}:} This benchmark is widely used in the literature on the MinDS problem~\cite{hedar2010hybrid, hedar2012simulated, potluri2011two, chalupa2018order, cai2021two, abed2022hybrid}. Graphs in this benchmark  simulate wireless sensor networks, in which  each vertex is a center of a circle. This benchmark contains 12 families, and each family contains 10 instances.
	
	\emph{T1/T2$^1$:} These two benchmarks consist of 53 families with a total of 530 instances, where each family contains 10 instances of the same size.
	
	
	\emph{BHOSLIB\footnote{\url{http://www.nlsde.buaa.edu.cn/~kexu/benchmarks/graph-benchmarks.htm}}:} This benchmark was translated from hard random SAT instances, which have been widely used to test heuristic algorithms \cite{xu2007benchmarks}.
	
	\emph{SNAP\footnote{\url{https://snap.stanford.edu/data/}}:} This benchmark is from the Stanford Large Network Dataset Collection and was used by Cai et al. \cite{cai2021two} to test MinDS algorithms. We collect all instances they used in the benchmark to evaluate DmDS.
	
	\emph{DIMACS10\footnote{\url{https://www.cc.gatech.edu/dimacs10/downloads.shtml}}:} These are synthetic and real-world instances from the 10th DIMACS implementation challenge on graph partitioning and graph clustering~\cite{DBLP:conf/dimacs/2012}.
	
	\emph{Network Repository\footnote{https://networkrepository.com}:} This benchmark is an interactive scientific network data repository \cite{nr},  which contains a large number of graphs of various categories and has been widely used for different graph  problems~\cite{cai2015balance,cai2021two,lamm2016finding}. 
	
	\subsection{Experiment Setup} \label{setup}
	All algorithms are implemented in C++ and compiled by gcc 7.1.0 with `-O3' option. All experiments are run under CentOS Linux release 7.6.1810 with Intel(R) Xeon(R) Gold 6254 CPU @ 3.10GHz and 128GB RAM. The parameters of FastDS are set to the same as those used by its original authors~\cite{cai2021two}, and the parameter $\alpha$ of DmDS is set to 0.5. Different settings of  $\alpha$  are studied in Section \ref{empirical}.  Regarding the parameter $t$ in the BMS heuristic, we follow the authors' suggestion (near 50) and set $t$ randomly in the range of 45 to 55.
	
	We compare DmDS with two state-of-the-art heuristic MinDS algorithms, ScBppw~\cite{fan2019efficient} and FastDS~\cite{cai2021two}.  ScBppw is designed to solve large graphs for MinDS, which combines a tabu strategy called score checking with best-picking with probabilistic walk and can obtain  smaller solutions than FastMWDS (proposed by Wang et al.~\cite{wang2018fast} in 2018) on many large graphs. FastDS presents a two-goal local search framework and uses inference rules to reduce search space, which is the best-performing algorithm at present. Both of the above algorithms are local search algorithms. The code for ScBppw is available online\footnote{\url{https://github.com/Fan-Yi}}, and the code for FastDS was kindly provided by its authors.
	
	Each algorithm runs 10 times for each instance with random seeds 1, 2, 3, ..., 10, and occupies a separate core. The cutoff time is set to 1,000 CPU seconds. For each instance, we report the best size (Min.) and the average size (Avg.) of the solutions found by each algorithm over the 10 runs. Specifically, for the UDG and T1/T2 benchmarks, we report the average $\overline{D_{min}}$ of the smallest solutions over all instances in a family.
	
	\subsection{Experimental Results}
	We report and discuss the experimental results obtained by  ScBppw, FastDS, and our DmDS algorithm. 
	Table \ref{tab:summary} gives a concise report on the competitive superiority of DmDS and the other two  algorithms on different benchmarks. Specifically, it shows the number of instances on which  an algorithm finds a smaller dominating set than other algorithms. As can be seen from Table \ref{tab:summary}, the  DmDS significantly outperforms the other two algorithms on all benchmarks except BHOSLIB, and performs not well only on one instance in  the BHOSLIB benchmark.
	
	\begin{table}[htbp] 
		\centering
		\caption{Summary results on all benchmarks}
		\resizebox{0.8\linewidth}{!}{
			\begin{tabular}{lrrr}
				\toprule
				Benchmark  & \multicolumn{1}{l}{DmDS} & \multicolumn{1}{l}{FastDS} & \multicolumn{1}{l}{ScBppw} \\
				\midrule
				UDG        & \textbf{1} & 0          & 0 \\
				T1/T2  & \textbf{8} & 1          & 0 \\
				BHOSLIB    & 0          & \textbf{1} & 0 \\
				SNAP       & \textbf{5} & 3          & 0 \\
				DIMACS10   & \textbf{12} & 5          & 0 \\
				Network Repository & \textbf{94} & 16         & 2 \\
				\midrule
				Total      & \textbf{120} & 26         & 2 \\
				\bottomrule
			\end{tabular}%
		}
		\label{tab:summary}%
	\end{table}%
	
	\begin{table}[h]
		\centering
		\caption{Experimental results on UDG}
		\resizebox{0.7\linewidth}{!}{
			\begin{tabular}{p{4em}rrr}
				\toprule
				Instance   & \multicolumn{1}{l}{DmDS} & \multicolumn{1}{l}{FastDS} & \multicolumn{1}{l}{ScBppw} \\
				Family     & \multicolumn{1}{l}{$\overline{D_{min}}$} & \multicolumn{1}{l}{$\overline{D_{min}}$} & \multicolumn{1}{l}{$\overline{D_{min}}$} \\
				\midrule
				V50U150    & \textbf{12.9} & \textbf{12.9} & \textbf{12.9} \\
				V50U200    & \textbf{9.4} & \textbf{9.4} & \textbf{9.4} \\
				V100U150   & \textbf{17.0} & \textbf{17.0} & 17.3  \\
				V100U200   & \textbf{10.4} & \textbf{10.4} & 10.6  \\
				V250U150   & \textbf{18.0} & \textbf{18.0} & 19.0  \\
				V250U200   & \textbf{10.8} & \textbf{10.8} & 11.5  \\
				V500U150   & \textbf{18.5} & \textbf{18.5} & 20.1  \\
				V500U200   & \textbf{11.2} & \textbf{11.2} & 12.4  \\
				V800U150   & \textbf{19.0} & \textbf{19.0} & 20.9  \\
				V800U200   & \textbf{11.7} & 11.8       & 12.6  \\
				V1000U150  & \textbf{19.1} & \textbf{19.1} & 21.3  \\
				V1000U200  & \textbf{12.0} & \textbf{12.0} & 13.0  \\
				\bottomrule
			\end{tabular}%
		}
		\label{tab:UDG}%
	\end{table}%
	
	\begin{table}[h]
		\centering
		\caption{Experimental results on T1/T2}
		\resizebox{\linewidth}{!}{
			\begin{tabular}{p{4em}rrr|p{4em}rrr}
				\toprule
				Instance   & \multicolumn{1}{l}{DmDS} & \multicolumn{1}{l}{FastDS} & \multicolumn{1}{l}{ScBppw} & \multicolumn{1}{l}{Instance} & \multicolumn{1}{l}{DmDS} & \multicolumn{1}{l}{FastDS} & \multicolumn{1}{l}{ScBppw} \\
				Family     & \multicolumn{1}{l}{$\overline{D_{min}}$} & \multicolumn{1}{l}{$\overline{D_{min}}$} & \multicolumn{1}{l}{$\overline{D_{min}}$} & \multicolumn{1}{l}{Family} & \multicolumn{1}{l}{$\overline{D_{min}}$} & \multicolumn{1}{l}{$\overline{D_{min}}$} & \multicolumn{1}{l}{$\overline{D_{min}}$} \\
				\midrule
				V50E50     & \textbf{17.0} & \textbf{17.0} & \textbf{17.0} & \multicolumn{1}{l}{V250E1000} & \textbf{36.0} & \textbf{36.0} & 36.6  \\
				V50E100    & \textbf{11.9} & \textbf{11.9} & 12.0       & \multicolumn{1}{l}{V250E2000} & \textbf{21.6} & 21.8       & 22.4  \\
				V50E250    & \textbf{6.0} & \textbf{6.0} & \textbf{6.0} & \multicolumn{1}{l}{V250E3000} & \textbf{16.1} & 16.2       & 16.6  \\
				V50E500    & \textbf{3.9} & \textbf{3.9} & \textbf{3.9} & \multicolumn{1}{l}{V250E5000} & \textbf{11.0} & \textbf{11.0} & 11.5  \\
				V50E750    & \textbf{3.0} & \textbf{3.0} & \textbf{3.0} & \multicolumn{1}{l}{V300E300} & \textbf{100.0} & \textbf{100.0} & 100.1  \\
				V100E100   & \textbf{33.6} & \textbf{33.6} & \textbf{33.6} & \multicolumn{1}{l}{V300E500} & \textbf{77.7} & \textbf{77.7} & 78.4  \\
				V100E250   & \textbf{19.9} & \textbf{19.9} & 20.1       & \multicolumn{1}{l}{V300E750} & \textbf{59.6} & \textbf{59.6} & 60.0  \\
				V100E500   & \textbf{12.2} & \textbf{12.2} & 12.5       & \multicolumn{1}{l}{V300E1000} & \textbf{48.6} & \textbf{48.6} & 49.5  \\
				V100E750   & \textbf{9.0} & 9.1        & 9.5        & \multicolumn{1}{l}{V300E2000} & \textbf{29.4} & \textbf{29.4} & 30.6  \\
				V100E1000  & \textbf{7.5} & \textbf{7.5} & 7.8        & \multicolumn{1}{l}{V300E3000} & \textbf{22.2} & \textbf{22.2} & 22.9  \\
				V100E2000  & \textbf{4.1} & \textbf{4.1} & 4.3        & \multicolumn{1}{l}{V300E5000} & \textbf{15.3} & 15.5       & 15.9  \\
				V150E150   & \textbf{50.0} & \textbf{50.0} & \textbf{50.0} & \multicolumn{1}{l}{V500E500} & \textbf{167.0} & \textbf{167.0} & \textbf{167.0} \\
				V150E250   & \textbf{39.1} & \textbf{39.1} & 39.3       & \multicolumn{1}{l}{V500E1000} & \textbf{114.7} & \textbf{114.7} & 116.3  \\
				V150E500   & \textbf{24.6} & \textbf{24.6} & 24.9       & \multicolumn{1}{l}{V500E2000} & \textbf{71.2} & \textbf{71.2} & 72.9  \\
				V150E750   & \textbf{18.3} & \textbf{18.3} & 18.8       & \multicolumn{1}{l}{V500E5000} & \textbf{37.0} & \textbf{37.0} & 38.7  \\
				V150E1000  & \textbf{15.0} & \textbf{15.0} & 15.4       & \multicolumn{1}{l}{V500E10000} & \textbf{22.6} & 22.7       & 23.3  \\
				V150E2000  & \textbf{9.0} & \textbf{9.0} & 9.5        & \multicolumn{1}{l}{V800E0} & \textbf{267.0} & \textbf{267.0} & \textbf{267.0} \\
				V150E3000  & \textbf{6.9} & \textbf{6.9} & \textbf{6.9} & \multicolumn{1}{l}{V800E1000} & \textbf{242.5} & \textbf{242.5} & 246.5  \\
				V200E250   & \textbf{61.1} & \textbf{61.1} & 61.7       & \multicolumn{1}{l}{V800E2000} & \textbf{158.3} & \textbf{158.3} & 162.1  \\
				V200E500   & \textbf{39.6} & \textbf{39.6} & 40.0       & \multicolumn{1}{l}{V800E5000} & \textbf{82.7} & \textbf{82.7} & 86.3  \\
				V200E750   & \textbf{30.0} & \textbf{30.0} & 30.4       & \multicolumn{1}{l}{V800E10000} & \textbf{50.7} & 50.8       & 53.1  \\
				V200E1000  & \textbf{24.4} & \textbf{24.4} & 25.0       & \multicolumn{1}{l}{V1000E1000} & \textbf{333.7} & \textbf{333.7} & \textbf{333.7} \\
				V200E2000  & \textbf{15.0} & \textbf{15.0} & 15.3       & \multicolumn{1}{l}{V1000E5000} & \textbf{121.1} & \textbf{121.1} & 126.0  \\
				V200E3000  & \textbf{11.0} & \textbf{11.0} & 11.3       & \multicolumn{1}{l}{V1000E10000} & 74.3       & \textbf{74.2} & 77.6  \\
				V250E250   & \textbf{83.3} & \textbf{83.3} & \textbf{83.3} & \multicolumn{1}{l}{V1000E15000} & \textbf{55.8} & 55.9       & 58.2  \\
				V250E500   & \textbf{57.8} & \textbf{57.8} & 58.1       & \multicolumn{1}{l}{V1000E20000} & \textbf{45.7} & 45.8       & 47.4  \\
				V250E750   & \textbf{44.0} & \textbf{44.0} & 44.9       &            &            &            &  \\
				\bottomrule
			\end{tabular}%
		}
		\label{tab:T1T2}%
	\end{table}%
	
	\begin{table}[htbp]
		\centering
		\caption{Experimental results on BHOSLIB}
		\resizebox{\linewidth}{!}{
			\begin{tabular}{lrrrrrrrr}
				\toprule
				\multicolumn{3}{l}{Instance}         & \multicolumn{2}{l}{DmDS} & \multicolumn{2}{l}{FastDS} & \multicolumn{2}{l}{ScBppw} \\
				\midrule
				Name       & \multicolumn{1}{l}{$|V|$} & \multicolumn{1}{l}{$|E|$} & \multicolumn{1}{l}{Avg.} & \multicolumn{1}{l}{Min.} & \multicolumn{1}{l}{Avg.} & 	\multicolumn{1}{l}{Min.} & \multicolumn{1}{l}{Avg.} & \multicolumn{1}{l}{Min.} \\
				\midrule
				frb40-19-1 & 760        & 41,314     & 14.0       & \textbf{14} & 14.1       & \textbf{14} & 15.8       & 15 \\
				frb40-19-2 & 760        & 41,263     & 14.3       & \textbf{14} & 14.4       & \textbf{14} & 16.5       & 16 \\
				frb40-19-3 & 760        & 41,095     & 15.0       & \textbf{15} & 15.0       & \textbf{15} & 16.5       & 16 \\
				frb40-19-4 & 760        & 41,605     & 14.3       & \textbf{14} & 14.5       & \textbf{14} & 16.2       & 15 \\
				frb40-19-5 & 760        & 41,619     & 14.2       & \textbf{14} & 14.2       & \textbf{14} & 15.2       & 15 \\
				frb45-21-1 & 945        & 59,186     & 16.2       & \textbf{16} & 16.3       & \textbf{16} & 18.0       & 17 \\
				frb45-21-2 & 945        & 58,624     & 16.2       & \textbf{16} & 16.5       & \textbf{16} & 18.1       & 18 \\
				frb45-21-3 & 945        & 58,245     & 16.1       & \textbf{16} & 16.0       & \textbf{16} & 17.8       & 17 \\
				frb45-21-4 & 945        & 58,549     & 16.2       & \textbf{16} & 16.2       & \textbf{16} & 18.0       & 17 \\
				frb45-21-5 & 945        & 58,579     & 16.0       & \textbf{16} & 16.2       & \textbf{16} & 18.2       & 18 \\
				frb50-23-1 & 1,150      & 80,072     & 18.1       & \textbf{18} & 18.1       & \textbf{18} & 19.9       & 19 \\
				frb50-23-2 & 1,150      & 80,851     & 18.0       & \textbf{18} & 18.2       & \textbf{18} & 19.8       & 19 \\
				frb50-23-3 & 1,150      & 81,068     & 18.1       & \textbf{18} & 18.1       & \textbf{18} & 20.0       & 20 \\
				frb50-23-4 & 1,150      & 80,258     & 18.2       & \textbf{18} & 18.0       & \textbf{18} & 20.0       & 19 \\
				frb50-23-5 & 1,150      & 80,035     & 18.1       & \textbf{18} & 18.2       & \textbf{18} & 19.9       & 19 \\
				frb53-24-1 & 1,272      & 94,227     & 19.0       & \textbf{19} & 19.0       & \textbf{19} & 21.1       & 21 \\
				frb53-24-2 & 1,272      & 94,289     & 19.4       & \textbf{19} & 19.2       & \textbf{19} & 21.2       & 21 \\
				frb53-24-3 & 1,272      & 94,127     & 19.0       & \textbf{19} & 19.0       & \textbf{19} & 20.4       & 20 \\
				frb53-24-4 & 1,272      & 94,308     & 18.6       & \textbf{18} & 18.4       & \textbf{18} & 20.2       & 19 \\
				frb53-24-5 & 1,272      & 94,226     & 19.0       & \textbf{19} & 19.2       & \textbf{19} & 21.5       & 20 \\
				frb56-25-1 & 1,400      & 109,676    & 20.0       & \textbf{20} & 20.1       & \textbf{20} & 22.2       & 21 \\
				frb56-25-2 & 1,400      & 109,401    & 20.1       & \textbf{20} & 20.2       & \textbf{20} & 22.5       & 21 \\
				frb56-25-3 & 1,400      & 109,379    & 20.1       & \textbf{20} & 20.1       & \textbf{20} & 22.2       & 21 \\
				frb56-25-4 & 1,400      & 110,038    & 20.1       & \textbf{20} & 20.2       & \textbf{20} & 22.5       & 22 \\
				frb56-25-5 & 1,400      & 109,601    & 19.9       & \textbf{19} & 19.9       & \textbf{19} & 21.8       & 21 \\
				frb59-26-1 & 1,534      & 126,555    & 20.8       & \textbf{20} & 20.7       & \textbf{20} & 22.6       & 22 \\
				frb59-26-2 & 1,534      & 126,163    & 21.0       & 21         & 20.8       & \textbf{20} & 22.6       & 22 \\
				frb59-26-3 & 1,534      & 126,082    & 21.2       & \textbf{21} & 21.0       & \textbf{21} & 23.1       & 23 \\
				frb59-26-4 & 1,534      & 127,011    & 21.3       & \textbf{21} & 21.1       & \textbf{21} & 23.6       & 23 \\
				frb59-26-5 & 1,534      & 125,982    & 21.3       & \textbf{21} & 21.2       & \textbf{21} & 24.0       & 23 \\
				frb100-40  & 4,000      & 572,774    & 36.4       & \textbf{36} & 36.3       & \textbf{36} & 40.1       & 39 \\
				\bottomrule
			\end{tabular}%
		}
		\label{tab:BHOSLIB}%
	\end{table}%
	
	\begin{table}[htbp]
		\centering
		\caption{Averaged convergence time on three standard benchmarks}
		\resizebox{0.8\linewidth}{!}{
			\tiny{
				\begin{tabular}{lrrr}
					\toprule
					Benchmark  & \multicolumn{1}{l}{DmDS} & \multicolumn{1}{l}{FastDS} & \multicolumn{1}{l}{ScBppw} \\
					\midrule
					UDG        & 0.199s     & 4.061s     & 0.002s \\
					T1/T2  & 4.508s     & 4.591s     & 0.001s \\
					BHOSLIB    & 38.236s    & 51.112s    & 0.003s \\
					\bottomrule
				\end{tabular}%
			}
		}
		\label{tab:avg_time_standard}%
	\end{table}%

	\subsubsection{Comparison on  Solution Quality}
	
	The results  on UDG, T1/T2, BHOSLIB, SNAP, and DIMACS10 benchmarks are reported in Tables \ref{tab:UDG}, \ref{tab:T1T2}, \ref{tab:BHOSLIB}, and \ref{tab:snap_and_dimacs10}, respectively. Since there are too many instances in the Network Repository benchmark, the results are separately reported in  Table~\ref{tab:Repository-1}, Table~\ref{tab:Repository-2},  Table~\ref{tab:Repository-3}, and Table~\ref{tab:Repository-4}. We mark the best solution found by all algorithms as bold. Regarding UDG and T1/T2 benchmarks in Tables~\ref{tab:UDG} and~\ref{tab:T1T2}, we follow Cai et al. ~\cite{cai2021two} and report the average best solution for each family. It can be seen that DmDS has the best performance in terms of the ``Min.'' values, FastDS has performance similar to DmDS in the ``Avg.'' values,  
	and ScBppw has the worst performance in both of these two values.  We now provide a more detailed discussion of the results. 
	

	
	Regarding  standard benchmarks, the DmDS, FastDS, and ScBppw algorithms obtain the best solutions respectively on  12 (all), 11, and two families in  the  UDG benchmark,  respectively on  52 (out of 53), 44, and 11 families in the T1/T2 benchmark, and respectively on 30, 31 (all), and zero instances in the  BHOSLIB benchmark.  In addition, DmDS obtains the best ``Avg.'' value on  21 (out of 31) BHOSLIB instances, outperforming other algorithms. These results indicate that DmDS outperforms the other two algorithms on the standard  benchmarks, though  FastDS can obtain a slightly smaller dominating set than DmDS on one BHOSLIB instance.
	
	Regarding large real-world sparse graphs, the DmDS, FastDS, and ScBppw algorithms obtain the best solutions respectively on 19 (out of 22), 17, and 11 instances in  the SNAP benchmark,  respectively on  22 (out of 27), 15, and two instances in the DIMACS10 benchmark, and  respectively on 183 (out of 201),  105, and 48 instances in the Network Repository benchmark. 
    In particular,  DmDS  obtains a smaller dominating set than other two algorithms on five SNAP benchmark instances, 12 DIMACS10 benchmark instances, and 94 Network Repository benchmark instances;  FastDS  obtains a smaller dominating set than other two algorithms on three SNAP benchmark instances,   five DIMACS10 benchmark instances, and 16 Network Repository benchmark instances; and ScBppw  obtains a smaller dominating set than other two algorithms only on two Network Repository benchmark instances. 
    Note that there are 78 brain network instances in Network Repository benchmark with vertex counts ranging from 300,000 to 1,000,000. Excluding these instances, DmDS, FastDS, and ScBppw obtain a smaller dominating set than the other two algorithms on 25, 12, and 2 Network Repository instances, respectively.
    Regarding the  ``Avg.'' value, DmDS obtains the best solutions on  17 (out of 22) SNAP instances, 19 (out of 27) DIMACS10 instances, and 178 (out of 201) Network Repository instances, which outperforms  ScBppw and FastDS  except for SNAP benchmark (FastDS can obtain the best solutions on  18 SNAP instances). These results suggest that the DmDS algorithm significantly outperforms the other two state-of-the-art algorithms on large sparse real-world benchmarks.

	\begin{table}[thbp]
		\centering
		\caption{Experimental results on SNAP and DIMACS10}
		\resizebox{\linewidth}{!}{
			\begin{tabular}{p{5.8em}rrrrrrrr}
				\toprule
				Instance   & \multicolumn{2}{l}{DmDS} & \multicolumn{2}{l}{FastDS} & \multicolumn{2}{l}{ScBppw} \\
				\cmidrule{2-7}    Name       & \multicolumn{1}{l}{Avg.} & \multicolumn{1}{l}{Min.} & \multicolumn{1}{l}{Avg.} & \multicolumn{1}{l}{Min.} & 	\multicolumn{1}{l}{Avg.} & \multicolumn{1}{l}{Min.} \\
				\midrule
				Wiki-Vote  & 1,116.0    & \textbf{1,116} & 1,116.0    & \textbf{1,116} & 1,116.0    & \textbf{1,116} \\
				p2p-*lla04 & 2,227.0    & \textbf{2,227} & 2,227.0    & \textbf{2,227} & 2,227.0    & \textbf{2,227} \\
				p2p-*lla25 & 4,519.0    & \textbf{4,519} & 4,519.0    & \textbf{4,519} & 4,519.0    & \textbf{4,519} \\
				cit-HepTh  & 1,011.0    & \textbf{1,011} & 1,011.0    & \textbf{1,011} & 1,019.6    & 1,018 \\
				p2p-*lla24 & 5,418.0    & \textbf{5,418} & 5,418.0    & \textbf{5,418} & 5,418.0    & \textbf{5,418} \\
				cit-HepPh  & 894.0      & \textbf{894} & 894.0      & \textbf{894} & 907.0      & 905 \\
				p2p-*lla30 & 7,169.0    & \textbf{7,169} & 7,169.0    & \textbf{7,169} & 7,169.3    & \textbf{7,169} \\
				p2p-*lla31 & 12,582.0   & \textbf{12,582} & 12,582.0   & \textbf{12,582} & 12,582.0   & \textbf{12,582} \\
				soc-Epinions1 & 15,734.0   & \textbf{15,734} & 15,734.0   & \textbf{15,734} & 15,734.0   & \textbf{15,734} \\
				Slashdot0811 & 14,312.0   & \textbf{14,312} & 14,312.0   & \textbf{14,312} & 14,312.0   & \textbf{14,312} \\
				Slashdot0902 & 15,305.0   & \textbf{15,305} & 15,305.0   & \textbf{15,305} & 15,305.0   & \textbf{15,305} \\
				amazon0302 & 35,599.7   & 35,593     & 35,597.2   & \textbf{35,587} & 36,210.7   & 36,178 \\
				email-EuAll & 17,976.0   & \textbf{17,976} & 17,976.0   & \textbf{17,976} & 17,976.0   & \textbf{17,976} \\
				web-Stanford & 13,199.2   & \textbf{13,197} & 13,198.3   & \textbf{13,197} & 13,320.9   & 13,312 \\
				web-*Dame & 23,733.8   & \textbf{23,733} & 23,735.0   & 23,735     & 23,747.1   & 23,745 \\
				amazon0312 & 45,488.3   & \textbf{45,480} & 45,490.2   & 45,487     & 45,936.4   & 45,908 \\
				amazon0601 & 42,289.9   & \textbf{42,283} & 42,287.8   & 42,284     & 42,730.3   & 42,708 \\
				amazon0505 & 47,314.9   & \textbf{47,309} & 47,315.4   & 47,310     & 47,743.9   & 47,718 \\
				web-BerkStan & 28,438.8   & 28,438     & 28,435.1   & \textbf{28,434} & 28,623.6   & 28,603 \\
				web-Google & 79,699.0   & 79,699     & 79,698.0   & \textbf{79,698} & 79,758.7   & 79,751 \\
				wiki-Talk  & 36,960.0   & \textbf{36,960} & 36,960.0   & \textbf{36,960} & 36,960.0   & \textbf{36,960} \\
				cit-Patents & 621,661.4  & \textbf{621,644} & 621,736.2  & 621,705    & 623,185.8  & 623,121 \\
				\midrule
				as-22july06 & 2,026.0    & \textbf{2,026} & 2,026.0    & \textbf{2,026} & 2,026.0    & \textbf{2,026} \\
				cond-*2005 & 5,664.0    & \textbf{5,664} & 5,664.0    & \textbf{5,664} & 5,690.9    & 5,686 \\
				kron\_*logn16 & 3,885.0    & \textbf{3,885} & 3,885.0    & \textbf{3,885} & 3,885.0    & \textbf{3,885} \\
				luxem*g\_osm & 37,753.2   & \textbf{37,751} & 37,751.5   & \textbf{37,751} & 38,444.8   & 38,421 \\
				rgg*\_17\_s0 & 12,319.3   & \textbf{12,317} & 12,332.7   & 12,329     & 13,848.9   & 13,827 \\
				wave       & 11,681.6   & 11,654     & 11,646.3   & \textbf{11,609} & 13,927.0   & 13,887 \\
				caida*Level & 40,523.0   & \textbf{40,523} & 40,523.0   & \textbf{40,523} & 40,580.4   & 40,573 \\
				coAut*seer & 33,197.0   & \textbf{33,197} & 33,197.0   & \textbf{33,197} & 33,358.1   & 33,343 \\
				citat*eseer & 43,412.0   & \textbf{43,412} & 43,412.0   & \textbf{43,412} & 43,434.3   & 43,430 \\
				coAut*DBLP & 43,978.0   & \textbf{43,978} & 43,978.0   & \textbf{43,978} & 44,094.9   & 44,084 \\
				cnr-2000   & 22,010.4   & \textbf{22,010} & 22,011.2   & \textbf{22,010} & 22,030.9   & 22,027 \\
				coPape*seer & 26,082.0   & \textbf{26,082} & 26,082.0   & \textbf{26,082} & 27,008.0   & 26,989 \\
				rgg*\_19\_s0 & 44,397.1   & \textbf{44,378} & 44,413.6   & 44,404     & 50,420.2   & 50,368 \\
				coPa*DBLP & 35,596.7   & \textbf{35,595} & 35,598.7   & 35,597     & 37,090.3   & 37,060 \\
				eu-2005    & 32,287.8   & \textbf{32,281} & 32,288.4   & 32,284     & 32,325.2   & 32,316 \\
				audikw\_1  & 9,486.2    & 9,469      & 9,467.3    & \textbf{9,459} & 10,568.5   & 10,534 \\
				ldoor      & 19,725.2   & \textbf{19,710} & 19,724.4   & 19,717     & 21,546.4   & 21,527 \\
				ecology1   & 202,057.3  & \textbf{201,491} & 201,664.9  & 201,505    & 238,250.4  & 238,052 \\
				rgg*\_20\_s0 & 84,690.6   & \textbf{84,665} & 84,708.7   & 84,693     & 96,675.2   & 96,556 \\
				in-2004    & 77,781.9   & \textbf{77,781} & 77,787.0   & 77,786     & 77,861.6   & 77,850 \\
				belgium\_osm & 468,925.8  & 468,909    & 468,854.8  & \textbf{468,846} & 476,811.0  & 476,674 \\
				rgg*\_21\_s0 & 162,242.7  & \textbf{162,170} & 162,303.2  & 162,273    & 185,435.2  & 185,346 \\
				rgg*\_22\_s0 & 312,553.2  & \textbf{312,471} & 312,886.3  & 312,621    & 356,871.8  & 356,750 \\
				cage15     & 457,967.5  & 457,686    & 456,650.6  & \textbf{456,524} & 476,883.6  & 476,574 \\
				rgg*\_23\_s0 & 608,116.9  & \textbf{607,682} & 609,133.6  & 608,149    & 687,953.9  & 687,850 \\
				rgg*\_24\_s0 & 1,217,145.7  & 1,209,519  & 1,207,476.8  & \textbf{1,205,162} & 1,328,143.5  & 1,327,935 \\
				uk-2002    & 1,038,649.7  & \textbf{1,038,641} & 1,038,674.3  & 1,038,662  & 1,040,090.5  & 1,040,064 \\
				\bottomrule
			\end{tabular}%
		}
		\label{tab:snap_and_dimacs10}%
	\end{table}%
	
	\begin{table}[htbp]
		\renewcommand{\thetable}{VII-I}
		\centering
		\caption{Experimental results on part I of Network Repository}
		\resizebox{\linewidth}{!}{
			\begin{tabular}{lrrrrrrrr}
				\toprule
				Instance   & \multicolumn{2}{l}{DmDS} & \multicolumn{2}{l}{FastDS} & \multicolumn{2}{l}{ScBppw} \\
				\cmidrule{2-7}    Name       & \multicolumn{1}{l}{Avg.} & \multicolumn{1}{l}{Min.} & \multicolumn{1}{l}{Avg.} & \multicolumn{1}{l}{Min.} & 	\multicolumn{1}{l}{Avg.} & \multicolumn{1}{l}{Min.} \\
				\midrule
				bio-yeast  & 353.0      & \textbf{353} & 353.0      & \textbf{353} & 353.0      & \textbf{353} \\
				bio*-protein-inter & 514.0      & \textbf{514} & 514.0      & \textbf{514} & 514.0      & \textbf{514} \\
				bio*fission-yeast & 280.0      & \textbf{280} & 280.0      & \textbf{280} & 280.0      & \textbf{280} \\
				bio-CE-GN  & 195.0      & \textbf{195} & 195.0      & \textbf{195} & 195.7      & \textbf{195} \\
				bio-HS-HT  & 456.0      & \textbf{456} & 456.0      & \textbf{456} & 456.6      & \textbf{456} \\
				bio-CE-HT  & 690.0      & \textbf{690} & 690.0      & \textbf{690} & 690.0      & \textbf{690} \\
				bio-DM-HT  & 734.0      & \textbf{734} & 734.0      & \textbf{734} & 734.4      & \textbf{734} \\
				bio-DR-CX  & 310.0      & \textbf{310} & 310.0      & \textbf{310} & 311.2      & \textbf{310} \\
				bio-grid-worm & 578.0      & \textbf{578} & 578.0      & \textbf{578} & 578.0      & \textbf{578} \\
				bio-DM-CX  & 513.0      & \textbf{513} & 513.0      & \textbf{513} & 513.5      & \textbf{513} \\
				bio-HS-LC  & 380.0      & \textbf{380} & 380.0      & \textbf{380} & 380.0      & \textbf{380} \\
				bio-HS-CX  & 495.0      & \textbf{495} & 495.0      & \textbf{495} & 495.0      & \textbf{495} \\
				ca-Erdos992 & 440.0      & \textbf{440} & 440.0      & \textbf{440} & 440.0      & \textbf{440} \\
				bio-grid-yeast & 286.0      & \textbf{286} & 286.0      & \textbf{286} & 286.0      & \textbf{286} \\
				bio-grid-fruitfly & 1,522.0    & \textbf{1,522} & 1,522.0    & \textbf{1,522} & 1,522.0    & \textbf{1,522} \\
				bio-dmela  & 1,453.0    & \textbf{1,453} & 1,453.0    & \textbf{1,453} & 1,453.0    & \textbf{1,453} \\
				bio-grid-human & 1,785.0    & \textbf{1,785} & 1,785.0    & \textbf{1,785} & 1,785.0    & \textbf{1,785} \\
				Oregon-1   & 992.0      & \textbf{992} & 992.0      & \textbf{992} & 992.0      & \textbf{992} \\
				Oregon-2   & 961.0      & \textbf{961} & 961.0      & \textbf{961} & 961.0      & \textbf{961} \\
				skirt      & 909.0      & \textbf{909} & 909.0      & \textbf{909} & 1,121.9    & 1,111 \\
				cbuckle    & 290.8      & \textbf{290} & 291.8      & 291        & 368.7      & 363 \\
				cyl6       & 241.0      & \textbf{241} & 241.0      & \textbf{241} & 308.0      & 301 \\
				bio-human-gene2 & 392.0      & \textbf{392} & 392.0      & \textbf{392} & 392.6      & \textbf{392} \\
				case9      & 2,403.0    & \textbf{2,403} & 2,403.0    & \textbf{2,403} & 2,403.0    & \textbf{2,403} \\
				bio-CE-CX  & 2,549.0    & \textbf{2,549} & 2,549.0    & \textbf{2,549} & 2,552.7    & 2,552 \\
				Dubcova1   & 1,024.0    & \textbf{1,024} & 1,024.0    & \textbf{1,024} & 1,183.2    & 1,175 \\
				olafu      & 285.0      & \textbf{285} & 285.0      & \textbf{285} & 364.1      & 349 \\
				bio-WormNet-v3 & 2,072.4    & \textbf{2,072} & 2,072.1    & \textbf{2,072} & 2,075.8    & 2,073 \\
				ca-AstroPh & 2,055.0    & \textbf{2,055} & 2,055.0    & \textbf{2,055} & 2,073.0    & 2,068 \\
				raefsky4   & 322.0      & \textbf{322} & 322.0      & \textbf{322} & 404.1      & 393 \\
				raefsky3   & 300.0      & \textbf{300} & 300.0      & \textbf{300} & 375.4      & 369 \\
				ca-CondMat & 2,990.0    & \textbf{2,990} & 2,990.0    & \textbf{2,990} & 3,011.9    & 3,008 \\
				bio-human-gene1 & 825.3      & \textbf{825} & 825.8      & \textbf{825} & 827.1      & 826 \\
				rgg\_n\_2\_15\_s0 & 3,473.4    & \textbf{3,472} & 3,475.8    & 3,475      & 3,857.2    & 3,838 \\
				bio-pdb1HYS & 413.5      & \textbf{412} & 422.5      & 421        & 501.2      & 497 \\
				c-62ghs    & 15,161.0   & \textbf{15,161} & 15,161.0   & \textbf{15,161} & 15,248.2   & 15,240 \\
				bio-mouse-gene & 2,713.6    & \textbf{2,713} & 2,713.2    & \textbf{2,713} & 2,715.9    & 2,714 \\
				c-66b      & 21,183.0   & \textbf{21,183} & 21,183.0   & \textbf{21,183} & 21,224.8   & 21,220 \\
				sc-nasasrb & 1,044.6    & \textbf{1,044} & 1,045.0    & 1,045      & 1,385.7    & 1,374 \\
				Dubcova2   & 4,096.0    & \textbf{4,096} & 4,096.1    & \textbf{4,096} & 4,709.8    & 4,698 \\
				rgg\_n\_2\_16\_s0 & 6,536.7    & \textbf{6,535} & 6,543.5    & 6,539      & 7,287.3    & 7,270 \\
				sc-pkustk11 & 2,260.0    & \textbf{2,260} & 2,260.3    & \textbf{2,260} & 2,367.6    & 2,361 \\
				sc-pkustk13 & 1,223.2    & \textbf{1,222} & 1,223.8    & \textbf{1,222} & 1,462.3    & 1,449 \\
				soc-buzznet & 127.0      & \textbf{127} & 127.0      & \textbf{127} & 127.0      & \textbf{127} \\
				soc-LiveMocha & 1,424.0    & \textbf{1,424} & 1,424.0    & \textbf{1,424} & 1,424.0    & \textbf{1,424} \\
				kron\_g500-logn17 & 7,468.0    & \textbf{7,468} & 7,468.0    & \textbf{7,468} & 7,468.0    & \textbf{7,468} \\
				web-uk-2005 & 1,421.0    & \textbf{1,421} & 1,421.0    & \textbf{1,421} & 1,421.0    & \textbf{1,421} \\
				sc-shipsec1 & 7,668.9    & 7,664      & 7,666.3    & \textbf{7,660} & 8,805.2    & 8,788 \\
				\bottomrule
			\end{tabular}%
		}
		\label{tab:Repository-1}%
	\end{table}%
	
	\begin{table}[htbp]
		\renewcommand{\thetable}{VII-II}
		\centering
		\caption{Experimental results on part II of Network Repository}
		\resizebox{\linewidth}{!}{
			\begin{tabular}{lrrrrrrrr}
				\toprule
				Instance   & \multicolumn{2}{l}{DmDS} & \multicolumn{2}{l}{FastDS} & \multicolumn{2}{l}{ScBppw} \\
				\cmidrule{2-7}    Name       & \multicolumn{1}{l}{Avg.} & \multicolumn{1}{l}{Min.} & \multicolumn{1}{l}{Avg.} & \multicolumn{1}{l}{Min.} & 	\multicolumn{1}{l}{Avg.} & \multicolumn{1}{l}{Min.} \\
				\midrule
				Dubcova3   & 4,096.0    & \textbf{4,096} & 4,097.0    & \textbf{4,096} & 4,725.8    & 4,707 \\
				soc-douban & 8,364.0    & \textbf{8,364} & 8,364.0    & \textbf{8,364} & 8,364.0    & \textbf{8,364} \\
				web-arabic-2005 & 16,901.5   & \textbf{16,901} & 16,901.7   & \textbf{16,901} & 16,908.2   & 16,904 \\
				rec-dating & 11,734.0   & \textbf{11,734} & 11,734.3   & \textbf{11,734} & 11,736.4   & 11,735 \\
				sc-shipsec5 & 10,310.9   & 10,299     & 10,303.3   & \textbf{10,296} & 11,879.1   & 11,854 \\
				soc-academia & 28,325.0   & \textbf{28,325} & 28,325.0   & \textbf{28,325} & 28,332.1   & 28,329 \\
				kron\_g500-logn18 & 14,268.0   & \textbf{14,268} & 14,268.0   & \textbf{14,268} & 14,268.0   & \textbf{14,268} \\
				sc-pwtk    & 4,198.0    & \textbf{4,196} & 4,202.1    & 4,201      & 5,520.1    & 5,490 \\
				rec-libimseti-dir & 12,954.8   & \textbf{12,954} & 12,954.8   & \textbf{12,954} & 12,956.1   & \textbf{12,954} \\
				rgg\_n\_2\_18\_s0 & 23,354.0   & \textbf{23,349} & 23,374.1   & 23,367     & 26,427.0   & 26,407 \\
				ca-dblp-2012 & 46,138.0   & \textbf{46,138} & 46,138.0   & \textbf{46,138} & 46,264.6   & 46,254 \\
				bn*M87119044 & 6,964.8    & \textbf{6,961} & 6,964.5    & 6,962      & 7,492.0    & 7,481 \\
				bn*M87116523 & 6,589.4    & \textbf{6,586} & 6,591.2    & 6,590      & 7,117.0    & 7,104 \\
				ca-MathSciNet & 65,572.0   & \textbf{65,572} & 65,572.0   & \textbf{65,572} & 65,597.2   & 65,592 \\
				sc-msdoor  & 8,604.3    & 8,600      & 8,605.4    & \textbf{8,596} & 10,317.0   & 10,282 \\
				kron\_g500-logn19 & 27,748.0   & \textbf{27,748} & 27,748.0   & \textbf{27,748} & 27,748.0   & \textbf{27,748} \\
				soc-dogster & 26,249.0   & \textbf{26,249} & 26,249.0   & \textbf{26,249} & 26,252.5   & 26,251 \\
				bn*M87118347 & 7,861.5    & \textbf{7,855} & 7,861.3    & \textbf{7,855} & 8,580.9    & 8,547 \\
				soc-twitter-higgs & 14,689.5   & \textbf{14,689} & 14,689.5   & \textbf{14,689} & 14,690.5   & \textbf{14,689} \\
				soc-youtube & 89,732.0   & \textbf{89,732} & 89,732.0   & \textbf{89,732} & 89,734.9   & 89,733 \\
				web-it-2004 & 32,997.0   & \textbf{32,997} & 32,997.0   & \textbf{32,997} & 32,999.4   & 32,998 \\
				soc-flickr & 98,062.0   & \textbf{98,062} & 98,062.0   & \textbf{98,062} & 98,064.5   & 98,063 \\
				soc-delicious & 55,722.0   & \textbf{55,722} & 55,722.0   & \textbf{55,722} & 55,728.3   & 55,726 \\
				ca-coauthors-dblp & 35,597.4   & \textbf{35,596} & 35,598.2   & 35,597     & 37,092.5   & 37,069 \\
				bn*M87110650 & 26,584.7   & \textbf{26,579} & 26,588.7   & 26,583     & 28,651.1   & 28,597 \\
				soc-FourSquare & 60,979.0   & \textbf{60,979} & 60,979.0   & \textbf{60,979} & 60,979.0   & \textbf{60,979} \\
				bn*873*1-bg & 19,680.2   & \textbf{19,673} & 19,683.5   & 19,677     & 20,997.5   & 20,982 \\
				bn*M87124152 & 17,154.5   & \textbf{17,149} & 17,156.0   & 17,150     & 18,571.3   & 18,551 \\
				bn*896*2-bg & 21,322.9   & \textbf{21,316} & 21,327.6   & 21,318     & 22,897.0   & 22,861 \\
				bn*869*1-bg & 23,432.2   & \textbf{23,425} & 23,434.7   & 23,428     & 24,987.5   & 24,955 \\
				bn*M87116517 & 19,283.1   & \textbf{19,275} & 19,287.7   & \textbf{19,275} & 20,952.5   & 20,927 \\
				bn*914*1-bg & 25,186.2   & \textbf{25,181} & 25,186.1   & 25,183     & 26,809.0   & 26,788 \\
				bn*864*1-bg & 23,505.2   & \textbf{23,497} & 23,507.8   & 23,501     & 25,000.5   & 24,984 \\
				bn*M87124670 & 25,621.4   & \textbf{25,612} & 25,627.2   & 25,615     & 28,093.0   & 28,069 \\
				bn*878*1-bg & 23,361.2   & \textbf{23,355} & 23,362.2   & \textbf{23,355} & 24,930.1   & 24,906 \\
				bn*914*2 & 24,485.8   & \textbf{24,478} & 24,490.0   & 24,483     & 26,114.4   & 26,068 \\
				bn*889*1 & 23,261.2   & \textbf{23,256} & 23,262.7   & 23,260     & 24,865.7   & 24,848 \\
				bn*917*1 & 22,010.1   & \textbf{22,000} & 22,013.6   & 22,004     & 23,524.7   & 23,499 \\
				bn*M87104201 & 15,787.2   & \textbf{15,781} & 15,787.4   & \textbf{15,781} & 17,151.9   & 17,132 \\
				bn*916*1 & 23,081.3   & \textbf{23,075} & 23,083.8   & 23,078     & 24,740.8   & 24,696 \\
				bn*869*2-bg & 25,097.9   & \textbf{25,088} & 25,100.6   & 25,092     & 26,814.0   & 26,746 \\
				bn*890*2 & 23,809.4   & \textbf{23,802} & 23,811.0   & 23,803     & 25,377.2   & 25,361 \\
				bn*M87115663 & 18,897.3   & \textbf{18,889} & 18,899.7   & 18,891     & 20,344.6   & 20,323 \\
				bn*M87124029 & 26,998.8   & \textbf{26,988} & 27,006.5   & 26,995     & 29,432.9   & 29,406 \\
				bn*M87127186 & 17,676.4   & 17,672     & 17,678.8   & \textbf{17,671} & 19,173.1   & 19,153 \\
				bn*913*2 & 22,474.8   & \textbf{22,469} & 22,476.2   & 22,472     & 24,003.5   & 23,976 \\
				bn*868*1-bg & 24,490.2   & \textbf{24,480} & 24,492.9   & 24,486     & 26,108.8   & 26,093 \\
				bn*868*2-bg & 24,939.8   & \textbf{24,933} & 24,941.1   & 24,934     & 26,540.4   & 26,482 \\
				\bottomrule
			\end{tabular}%
		}
		\label{tab:Repository-2}%
	\end{table}%
	
	\begin{table}[htbp]
		\renewcommand{\thetable}{VII-III}
		\centering
		\caption{Experimental results on part III of Network Repository}
		\resizebox{0.9\linewidth}{!}{
			\begin{tabular}{p{6em}rrrrrrrr}
				\toprule
				Instance   & \multicolumn{2}{l}{DmDS} & \multicolumn{2}{l}{FastDS} & \multicolumn{2}{l}{ScBppw} \\
				\cmidrule{2-7}    Name       & \multicolumn{1}{l}{Avg.} & \multicolumn{1}{l}{Min.} & \multicolumn{1}{l}{Avg.} & \multicolumn{1}{l}{Min.} & 	\multicolumn{1}{l}{Avg.} & \multicolumn{1}{l}{Min.} \\
				\midrule
				bn*M87101698 & 21,064.1   & \textbf{21,057} & 21,066.7   & 21,061     & 22,736.2   & 22,709 \\
				bn*911*2 & 25,251.6   & \textbf{25,241} & 25,252.2   & 25,245     & 26,940.6   & 26,910 \\
				bn*M87105966 & 21,680.3   & \textbf{21,675} & 21,686.1   & 21,677     & 23,389.2   & 23,355 \\
				bn*865*1-bg & 21,631.0   & \textbf{21,621} & 21,631.8   & 21,622     & 23,122.4   & 23,097 \\
				bn*871*2-bg & 23,030.7   & \textbf{23,026} & 23,034.9   & 23,027     & 24,594.4   & 24,559 \\
				bn*M87104509 & 27,461.1   & \textbf{27,448} & 27,467.2   & 27,459     & 29,796.0   & 29,766 \\
				bn*867*1-bg & 25,429.7   & \textbf{25,423} & 25,431.2   & 25,426     & 27,115.4   & 27,090 \\
				bn*918*1 & 24,366.3   & 24,359     & 24,366.4   & \textbf{24,358} & 26,088.8   & 26,057 \\
				bn*M87108808 & 22,359.2   & \textbf{22,353} & 22,360.1   & 22,354     & 24,112.5   & 24,079 \\
				bn*M87125286 & 19,797.9   & \textbf{19,792} & 19,800.3   & 19,795     & 21,281.7   & 21,264 \\
				rec-epinion & 9,037.0    & \textbf{9,037} & 9,037.0    & \textbf{9,037} & 9,038.0    & 9,038 \\
				bn*M87110148 & 20,458.2   & \textbf{20,448} & 20,466.3   & 20,454     & 22,202.1   & 22,179 \\
				bn*M87110670 & 26,034.8   & 26,022     & 26,041.0   & \textbf{26,021} & 28,495.5   & 28,446 \\
				bn*M87125334 & 30,331.1   & \textbf{30,324} & 30,343.3   & 30,333     & 33,271.3   & 33,228 \\
				bn*874*2-bg & 25,406.0   & \textbf{25,400} & 25,409.4   & \textbf{25,400} & 27,106.0   & 27,090 \\
				soc-digg   & 66,155.0   & \textbf{66,155} & 66,155.0   & \textbf{66,155} & 66,156.2   & \textbf{66,155} \\
				bn*M87101705 & 18,646.9   & \textbf{18,638} & 18,642.6   & 18,639     & 20,267.4   & 20,239 \\
				bn*M87124563 & 16,617.1   & 16,609     & 16,617.2   & \textbf{16,608} & 18,212.9   & 18,182 \\
				bn*876*2-bg & 29,366.0   & \textbf{29,349} & 29,366.2   & 29,355     & 31,183.8   & 31,159 \\
				bn*M87125521 & 30,312.1   & \textbf{30,297} & 30,318.0   & 30,308     & 33,179.4   & 33,153 \\
				bn*886*1 & 24,542.1   & \textbf{24,528} & 24,539.2   & 24,534     & 26,228.3   & 26,182 \\
				bn*M87107242 & 31,874.1   & \textbf{31,862} & 31,882.8   & 31,876     & 34,356.0   & 34,326 \\
				bn*912*\_2 & 24,889.2   & \textbf{24,880} & 24,889.4   & 24,881     & 26,554.8   & 26,540 \\
				bn*M87113878 & 14,327.7   & \textbf{14,321} & 14,330.8   & 14,325     & 15,588.2   & 15,568 \\
				bn*M87121714 & 28,440.6   & \textbf{28,432} & 28,449.3   & 28,441     & 31,131.7   & 31,096 \\
				bn*876*\_1-bg & 29,380.3   & \textbf{29,370} & 29,387.7   & 29,382     & 31,254.5   & 31,218 \\
				bn*M87118219 & 23,066.9   & \textbf{23,062} & 23,074.7   & 23,064     & 24,902.6   & 24,875 \\
				bn*M87127677 & 23,412.1   & \textbf{23,401} & 23,412.7   & 23,403     & 25,081.4   & 25,052 \\
				bn*M87118759 & 28,597.4   & \textbf{28,587} & 28,604.9   & 28,597     & 31,307.9   & 31,270 \\
				kron\_*logn20 & 53,350.0   & \textbf{53,350} & 53,350.0   & \textbf{53,350} & 53,350.0   & \textbf{53,350} \\
				bn*870*\_1-bg & 27,821.9   & \textbf{27,811} & 27,824.6   & 27,817     & 29,553.7   & 29,531 \\
				bn*M87125691 & 31,566.3   & \textbf{31,552} & 31,575.0   & 31,567     & 34,417.3   & 34,384 \\
				bn*M87111392 & 23,242.1   & \textbf{23,219} & 23,241.2   & 23,231     & 25,009.9   & 24,985 \\
				bn*M87109786 & 29,334.7   & \textbf{29,324} & 29,339.7   & 29,331     & 32,097.8   & 32,072 \\
				bn*M87102230 & 28,964.7   & \textbf{28,957} & 28,972.1   & 28,965     & 31,551.3   & 31,497 \\
				bn*870*\_2-bg & 27,226.1   & \textbf{27,218} & 27,228.0   & 27,222     & 28,981.8   & 28,943 \\
				bn*M87128194 & 23,101.5   & \textbf{23,088} & 23,106.6   & 23,097     & 25,121.4   & 25,094 \\
				bn*M87129974 & 32,349.2   & \textbf{32,334} & 32,351.0   & 32,336     & 35,012.9   & 34,941 \\
				bn*M87125330 & 30,379.7   & \textbf{30,366} & 30,387.6   & 30,378     & 33,276.7   & 33,228 \\
				bn*M87119472 & 26,788.3   & \textbf{26,775} & 26,799.1   & 26,787     & 29,618.7   & 29,599 \\
				bn*M87103674 & 29,427.4   & \textbf{29,404} & 29,434.9   & 29,425     & 32,025.4   & 31,999 \\
				bn*M87118954 & 19,318.8   & \textbf{19,306} & 19,325.6   & 19,316     & 21,721.8   & 21,703 \\
				bn*M87123142 & 30,256.5   & \textbf{30,253} & 30,264.6   & 30,255     & 33,153.3   & 33,121 \\
				bn*M87125989 & 32,248.0   & \textbf{32,242} & 32,259.1   & 32,249     & 35,412.7   & 35,368 \\
				bn*M87104300 & 24,855.2   & \textbf{24,848} & 24,864.3   & 24,851     & 27,446.4   & 27,375 \\
				bn*M87127667 & 19,365.5   & \textbf{19,355} & 19,368.0   & 19,359     & 21,772.6   & 21,734 \\
				bn*M87128519 & 21,196.1   & \textbf{21,186} & 21,199.7   & 21,191     & 23,111.7   & 23,087 \\
				bn*M87113679 & 28,895.0   & \textbf{28,886} & 28,904.4   & 28,888     & 31,765.1   & 31,713 \\
				\bottomrule
			\end{tabular}%
			\label{tab:Repository-3}%
		}
	\end{table}%

	\begin{table}[htbp]
		\renewcommand{\thetable}{VII-IV}
		\centering
		\caption{Experimental results on part IV of Network Repository}
		\resizebox{\linewidth}{!}{
			\begin{tabular}{p{7.2em}rrrrrrrr}
				\toprule
				Instance   & \multicolumn{2}{l}{DmDS} & \multicolumn{2}{l}{FastDS} & \multicolumn{2}{l}{ScBppw} \\
				\cmidrule{2-7}    Name       & \multicolumn{1}{l}{Avg.} & \multicolumn{1}{l}{Min.} & \multicolumn{1}{l}{Avg.} & \multicolumn{1}{l}{Min.} & 	\multicolumn{1}{l}{Avg.} & \multicolumn{1}{l}{Min.} \\
				\midrule
				bn*M87115834 & 29,321.9   & \textbf{29,309} & 29,331.4   & 29,310     & 32,194.8   & 32,152 \\
				bn*M87117515 & 28,704.7   & \textbf{28,688} & 28,719.0   & 28,707     & 31,984.5   & 31,960 \\
				bn*M87123456 & 32,750.9   & \textbf{32,737} & 32,762.7   & 32,749     & 35,956.5   & 35,874 \\
				ca-IMDB    & 120,570.3  & \textbf{120,570} & 120,570.8  & \textbf{120,570} & 120,573.7  & 120,572 \\
				sc-ldoor   & 19,746.2   & 19,740     & 19,743.0   & \textbf{19,731} & 23,621.6   & 23,582 \\
				bn*M87122310 & 16,712.0   & \textbf{16,701} & 16,716.1   & 16,707     & 18,810.1   & 18,772 \\
				bn*M87102575 & 17,672.6   & \textbf{17,659} & 17,675.6   & 17,664     & 19,982.8   & 19,965 \\
				bn*M87117093 & 22,191.3   & \textbf{22,179} & 22,194.1   & 22,188     & 24,643.5   & 24,605 \\
				bn*M87126525 & 15,326.8   & \textbf{15,316} & 15,333.4   & 15,323     & 17,448.5   & 17,424 \\
				ca-holly*-2009 & 48,733.5   & \textbf{48,728} & 48,746.6   & 48,743     & 50,369.7   & 50,351 \\
				inf-roadNet-PA & 326,866.6  & \textbf{326,845} & 326,939.6  & 326,921    & 332,187.0  & 332,099 \\
				rt-retweet-crawl & 75,740.0   & \textbf{75,740} & 75,740.0   & \textbf{75,740} & 75,740.2   & \textbf{75,740} \\
				soc-you*snap & 213,122.0  & \textbf{213,122} & 213,122.0  & \textbf{213,122} & 213,123.0  & \textbf{213,122} \\
				soc-lastfm & 67,226.0   & \textbf{67,226} & 67,226.0   & \textbf{67,226} & 67,226.0   & \textbf{67,226} \\
				kron*-logn21 & 102,238.0  & \textbf{102,238} & 102,238.0  & \textbf{102,238} & 102,238.0  & \textbf{102,238} \\
				soc-pokec  & 207,311.2  & 207,307    & 207,308.7  & \textbf{207,303} & 207,394.9  & 207,385 \\
				tech-as-skitter & 181,718.0  & \textbf{181,717} & 181,717.9  & \textbf{181,717} & 181,867.8  & 181,856 \\
				soc-flickr-und & 295,700.0  & \textbf{295,700} & 295,700.0  & \textbf{295,700} & 295,704.8  & 295,702 \\
				web-wiki*2009 & 346,582.3  & \textbf{346,580} & 346,583.9  & 346,582    & 346,684.3  & 346,670 \\
				web-wiki*growth & 116,814.5  & \textbf{116,814} & 116,814.6  & \textbf{116,814} & 116,819.7  & 116,817 \\
				inf-roadNet-CA & 586,395.0  & \textbf{586,364} & 586,492.8  & 586,481    & 595,852.8  & 595,760 \\
				web-baidu-baike & 276,745.0  & \textbf{276,745} & 276,745.0  & \textbf{276,745} & 276,748.3  & 276,747 \\
				tech-ip    & 154.9      & \textbf{154} & 154.8      & \textbf{154} & 156.7      & 155 \\
				soc-flixster & 91,019.0   & \textbf{91,019} & 91,019.0   & \textbf{91,019} & 91,019.0   & \textbf{91,019} \\
				memchip    & 464,234.4  & \textbf{464,194} & 464,263.9  & 464,215    & 466,616.9  & 466,512 \\
				socfb-B-anon & 187,030.0  & \textbf{187,030} & 187,030.0  & \textbf{187,030} & 187,030.2  & \textbf{187,030} \\
				soc-orkut  & 110,566.1  & \textbf{110,550} & 110,599.3  & 110,578    & 111,107.2  & 111,080 \\
				soc-orkut-dir & 93,619.5   & \textbf{93,603} & 93,655.4   & 93,630     & 94,005.9   & 93,982 \\
				socfb-A-anon & 201,690.0  & \textbf{201,690} & 201,690.0  & \textbf{201,690} & 201,690.3  & \textbf{201,690} \\
				wiki*link\_en & 212,876.0  & \textbf{212,876} & 212,876.0  & \textbf{212,876} & 212,876.0  & \textbf{212,876} \\
				Freescale1 & 597,141.5  & \textbf{597,132} & 597,162.7  & 597,152    & 598,622.1  & 598,580 \\
				patents    & 632,687.8  & \textbf{632,665} & 632,776.5  & 632,733    & 634,033.5  & 633,989 \\
				dblp-author & 1,411,236.0  & \textbf{1,411,236} & 1,411,236.3  & \textbf{1,411,236} & 1,411,237.5  & \textbf{1,411,236} \\
				delicious-ti & 140,676.0  & \textbf{140,676} & 140,676.0  & \textbf{140,676} & 140,676.4  & \textbf{140,676} \\
				soc-livejournal & 793,887.9  & \textbf{793,887} & 793,888.6  & 793,888    & 793,994.6  & 793,987 \\
				delaunay\_n22 & 641,622.1  & 641,480    & 641,542.6  & \textbf{641,394} & 688,804.3  & 688,628 \\
				channel*b050 & 343,731.6  & 343,328    & 343,003.4  & \textbf{342,733} & 392,503.0  & 392,415 \\
				ljournal-2008 & 1,005,858.0  & \textbf{1,005,858} & 1,005,858.0  & \textbf{1,005,858} & 1,005,986.5  & 1,005,974 \\
				soc-ljour*2008 & 1,005,858.0  & \textbf{1,005,858} & 1,005,858.0  & \textbf{1,005,858} & 1,005,980.1  & 1,005,970 \\
				sc-rel9    & 119,305.4  & 119,234    & 119,247.4  & \textbf{119,208} & 130,424.5  & 123,108 \\
				indochina-2004 & 399,928.9  & \textbf{399,922} & 399,952.9  & 399,949    & 400,516.5  & 400,507 \\
				soc-live*groups & 1,071,123.0  & \textbf{1,071,123} & 1,071,123.0  & \textbf{1,071,123} & 1,071,124.0  & \textbf{1,071,123} \\
				delaunay\_n23 & 1,294,633.7  & \textbf{1,291,532} & 1,295,994.0  & 1,293,999  & 1,377,716.0  & 1,377,473 \\
				friendster & 656,463.0  & \textbf{656,463} & 656,463.0  & \textbf{656,463} & 656,466.4  & 656,464 \\
				wb-edu     & 833,550.9  & \textbf{833,546} & 833,557.0  & 833,553    & 834,338.9  & 834,307 \\
				twitter\_mpi & 566,313.0  & \textbf{566,313} & 566,313.0  & \textbf{566,313} & 566,313.5  & \textbf{566,313} \\
				inf-germany\_os & 3,815,331.6  & 3,813,153  & 3,800,840.1  & \textbf{3,800,262} & 3,847,240.5  & 3,846,783 \\
				hugetrace-00010 & 3,301,282.3  & 3,274,050  & 3,256,924.8  & \textbf{3,253,894} & 3,392,443.6  & 3,391,758 \\
				dbpedia-link & 1,968,836.0  & \textbf{1,968,836} & 1,968,837.9  & 1,968,837  & 1,968,839.0  & 1,968,839 \\
				inf-road\_central & 4,581,139.1  & \textbf{4,580,173} & 4,585,195.9  & 4,584,680  & 4,615,658.1  & 4,615,457 \\
				road\_central & 4,580,701.5  & \textbf{4,579,690} & 4,583,212.3  & 4,582,927  & 4,615,769.5  & 4,615,468 \\
				hugetrace-*20 & 4,501,029.6  & 4,496,071  & 4,428,935.2  & \textbf{4,412,314} & 4,509,342.4  & 4,508,908 \\
				delaunay\_n24 & 2,690,868.0  & 2,672,734  & 2,661,379.5  & \textbf{2,651,597} & 2,755,337.8  & 2,755,004 \\
				hugebubs-*20 & 6,082,340.5  & 6,037,152  & 6,034,858.7  & 6,028,600  & 5,985,251.5  & \textbf{5,984,831} \\
				inf-road-usa & 7,817,568.8  & \textbf{7,804,560} & 7,826,434.1  & 7,818,748  & 7,855,036.8  & 7,854,395 \\
				inf-europe\_os & 17,385,174.3  & 17,359,873 & 17,205,514.2  & 17,191,561 & 17,019,317.7  & \textbf{17,017,316} \\
				socfb-uci-uni & 865,675.0  & \textbf{865,675} & 865,675.0  & \textbf{865,675} & 865,675.0  & \textbf{865,675} \\
				\bottomrule
			\end{tabular}%
			\label{tab:Repository-4}%
		}
	\end{table}%
	
	\subsubsection{Comparison of Convergence Time}

	We evaluate the convergence speed of DmDS by comparing the convergence time for which  an algorithm  finds the final solution for each instance. 
	Regarding the standard benchmarks, since the size of the instances is small, all of the three algorithms have a short convergence time. Here we report the averaged convergence time of the  three algorithms on these benchmarks. See Table~\ref{tab:avg_time_standard} for the results: While ScBppw has the shortest convergence time (due to the lack of further improvement after the solution is obtained), DmDS has a  faster convergence time than FastDS.  
	
	Table~\ref{tab:time} reports the convergence time of the three algorithms for the large sparse real-world instances,  on which at least two algorithms obtain the smallest  ``Min.'' value,  in total 113 instances. For each of such instances, the shortest convergence time among the three algorithms is highlighted in bold, and we use the notation ``N/A'' to indicate that the corresponding algorithm did not find a best solution (the smallest ``Min.'' value).  In addition, the running time less than $10^{-4}$ is reported as 0. As shown in the table,  DmDS has a shorter convergence time than other two algorithms on 41 instances,  while  FastDS and ScBppw have this advantage on  29 and 43 instances, respectively.  These results indicate that  in terms of the convergence time, ScBppw performs slightly better than DmDS  and DmDS performs better than FastDS. Note that ScBppw can not find the best solution on 54 (out of 113) instances, and both DmDS and FastDS can find a best solution on  the 113 instances. This indicates that DmDS is superior to the other two algorithms on these instances. 
	
	\begin{table}[htbp]
		\setcounter{table}{7}
		\centering
		\caption{Averaged convergence time for instances with ties}
		\resizebox{\linewidth}{!}{
			\begin{tabular}{p{7em}rrr|p{7em}rrr}
				\toprule
				Instance   & \multicolumn{1}{l}{DmDS} & \multicolumn{1}{l}{FastDS} & \multicolumn{1}{l|}{ScBppw} & c-62ghs    & 0.099      & \textbf{0.050} & N/A \\
				\cmidrule{1-4}    Wiki-Vote  & 0.032      & 0.020      & \textbf{0.008} & bio-mouse-gene & \textbf{91.292} & 185.752    & N/A \\
				p2p-Gnutella04 & 0.039      & 0.057      & \textbf{0.038} & c-66b      & 0.074      & \textbf{0.048} & N/A \\
				p2p-Gnutella25 & 0.029      & 0.017      & \textbf{0.005} & Dubcova2   & \textbf{0.357} & 0.370      & N/A \\
				cit-HepTh  & 0.603      & \textbf{0.444} & N/A        & sc-pkustk11 & 599.457    & \textbf{441.789} & N/A \\
				p2p-Gnutella24 & 0.039      & \textbf{0.032} & 0.048      & sc-pkustk13 & \textbf{186.771} & 443.232    & N/A \\
				cit-HepPh  & \textbf{297.149} & 336.827    & N/A        & soc-buzznet & 794.951    & 628.717    & \textbf{0.087} \\
				p2p-Gnutella30 & 0.052      & 0.039      & \textbf{0.019} & soc-LiveMocha & 0.847      & 0.418      & \textbf{0.114} \\
				p2p-Gnutella31 & 0.082      & 0.063      & \textbf{0.051} & kron\_*logn17 & 1.586      & \textbf{1.251} & 1.894 \\
				soc-Epinions1 & 0.152      & 0.098      & \textbf{0.071} & web-uk-2005 & 0.504      & 0.491      & \textbf{0.191} \\
				soc-Slash*0811 & 0.177      & \textbf{0.129} & 0.176      & Dubcova3   & \textbf{13.442} & 198.133    & N/A \\
				soc-Slash*0902 & 0.208      & \textbf{0.186} & 0.25       & soc-douban & 0.097      & 0.088      & \textbf{0.055} \\
				email-EuAll & 0.126      & 0.116      & \textbf{0.057} & web-arabic-2005 & 172.160    & \textbf{39.469} & N/A \\
				web-Stanford & 584.936    & \textbf{573.480} & N/A        & rec-dating & \textbf{25.436} & 196.068    & N/A \\
				wiki-Talk  & 2.854      & 2.640      & \textbf{1.376} & soc-academia & \textbf{2.854} & 3.342      & N/A \\
				\cmidrule{1-4}    as-22july06 & 0.011      & 0.015      & \textbf{0.006} & kron\_*logn18 & 4.025      & 2.087      & \textbf{0.709} \\
				cond-mat-2005 & \textbf{0.254} & 0.561      & N/A        & rec-libimseti-dir & \textbf{75.206} & 443.945    & 248.972 \\
				kron\_*logn16 & 1.712      & 0.894      & \textbf{0.236} & ca-dblp-2012 & \textbf{4.947} & 12.918     & N/A \\
				luxem*\_osm & 663.455    & \textbf{537.778} & N/A        & ca-MathSciNet & \textbf{0.919} & 1.424      & N/A \\
				caida*Level & \textbf{52.341} & 120.622    & N/A        & kron\_*logn19 & 9.599      & 5.074      & \textbf{1.759} \\
				coAuthors*seer & \textbf{1.375} & 3.094      & N/A        & soc-dogster & \textbf{7.428} & 9.708      & N/A \\
				citation*seer & \textbf{16.332} & 20.198     & N/A        & bn*M87118347 & 926.687    & \textbf{892.579} & N/A \\
				coAuth*DBLP & \textbf{1.192} & 3.308      & N/A        & soc-twitter-higgs & \textbf{384.006} & 427.345    & 717.939 \\
				cnr-2000   & \textbf{27.705} & 94.149     & N/A        & soc-youtube & 1.579      & \textbf{1.493} & N/A \\
				coPapers*seer & \textbf{353.499} & 433.597    & N/A        & web-it-2004 & 1.169      & \textbf{0.664} & N/A \\
				\cmidrule{1-4}    bio-yeast  & 0.001      & 0.001      & \textbf{0} & soc-flickr & \textbf{5.324} & 11.946     & N/A \\
				bio*protein-inter & 0.001      & 0.001      & \textbf{0} & soc-delicious & \textbf{1.088} & 1.392      & N/A \\
				bio*fission-yeast & 0.005      & 0.004      & \textbf{0} & soc-FourSquare & 2.181      & 1.507      & \textbf{1.152} \\
				bio-CE-GN  & 0.020      & 0.012      & \textbf{0.009} & bn*M87116517 & \textbf{972.049} & 981.889    & N/A \\
				bio-HS-HT  & 0.006      & 0.005      & \textbf{0} & bn*878*\_1-bg & 969.877    & \textbf{968.036} & N/A \\
				bio-CE-HT  & 0.001      & 0.001      & \textbf{0} & bn*M87104201 & 972.713    & \textbf{945.370} & N/A \\
				bio-DM-HT  & 0.005      & 0.004      & \textbf{0.002} & rec-epinion & 10.045     & \textbf{5.271} & N/A \\
				bio-DR-CX  & \textbf{0.360} & 0.667      & 2.071      & bn*874*\_2-bg & 976.101    & \textbf{975.953} & N/A \\
				bio-grid-worm & 0.002      & 0.003      & \textbf{0} & soc-digg   & 3.454      & 2.326      & \textbf{2.08} \\
				bio-DM-CX  & \textbf{0.107} & 0.136      & 8.192      & kron\_*logn20 & 25.652     & 14.378     & \textbf{4.697} \\
				bio-HS-LC  & \textbf{0.013} & 0.016      & 0.03       & ca-IMDB    & 241.239    & \textbf{228.911} & N/A \\
				bio-HS-CX  & 0.218      & \textbf{0.142} & 2.076      & rt-retweet-crawl & 2.863      & 1.663      & \textbf{1.334} \\
				ca-Erdos992 & 0.002      & 0.003      & \textbf{0.002} & soc-you*-snap & 3.140      & \textbf{2.632} & 3.729 \\
				bio-grid-yeast & 0.057      & \textbf{0.047} & 0.062      & soc-lastfm & 2.079      & 2.326      & \textbf{1.357} \\
				bio-grid-fruitfly & 0.007      & 0.008      & \textbf{0.002} & kron\_*logn21 & 78.569     & 41.929     & \textbf{19.113} \\
				bio-dmela  & 0.004      & 0.007      & \textbf{0.004} & tech-as-skitter & \textbf{489.562} & 720.781    & N/A \\
				bio-grid-human & 0.016      & 0.013      & \textbf{0.004} & soc-flickr-und & 10.337     & \textbf{7.794} & N/A \\
				\cmidrule{1-4}    Oregon-1   & 0.008      & 0.006      & \textbf{0.003} & web-wiki*owth & \textbf{486.888} & 550.193    & N/A \\
				Oregon-2   & 0.010      & 0.009      & \textbf{0.002} & web-baidu-baike & 25.677     & \textbf{23.667} & N/A \\
				skirt      & \textbf{3.510} & 7.144      & N/A        & tech-ip    & 611.045    & \textbf{293.682} & N/A \\
				cyl6       & 0.104      & \textbf{0.084} & N/A        & soc-flixster & 3.450      & 4.113      & \textbf{2.195} \\
				bio-hu*gene2 & \textbf{48.279} & 84.490     & 296.277    & socfb-B-anon & 24.674     & 15.172     & \textbf{12.04} \\
				case9      & 17.373     & 51.687     & \textbf{0.163} & socfb-A-anon & 30.082     & \textbf{27.239} & 41.825 \\
				bio-CE-CX  & \textbf{36.664} & 45.116     & N/A        & wikipedia\_link\_en & 38.314     & 22.601     & \textbf{22.081} \\
				Dubcova1   & \textbf{0.208} & 0.365      & N/A        & dblp-author & \textbf{16.355} & 347.832    & 639.67 \\
				olafu      & \textbf{2.789} & 4.428      & N/A        & delicious-ti & 31.238     & 20.972     & \textbf{7.518} \\
				bio-WormNet-v3 & \textbf{194.875} & 390.686    & N/A        & ljournal-2008 & \textbf{187.721} & 445.812    & N/A \\
				ca-AstroPh & \textbf{0.171} & 0.270      & N/A        & soc-ljou*-2008 & \textbf{127.607} & 422.618    & N/A \\
				raefsky4   & \textbf{0.430} & 0.554      & N/A        & soc*groups & 190.952    & \textbf{126.396} & 300.159 \\
				raefsky3   & \textbf{0.162} & 0.293      & N/A        & friendster & 147.052    & \textbf{80.067} & N/A \\
				ca-CondMat & \textbf{0.986} & 1.784      & N/A        & twitter\_mpi & 140.403    & 72.307     & \textbf{56.794} \\
				bio-hu*gene1 & 252.077    & \textbf{153.672} & N/A        & socfb-uci-uni & 409.629    & 237.207    & \textbf{70.993} \\
				\bottomrule
			\end{tabular}%
		}
		\label{tab:time}%
	\end{table}%

    \begin{figure}[htbp]
		
		\includegraphics[width=0.9\linewidth]{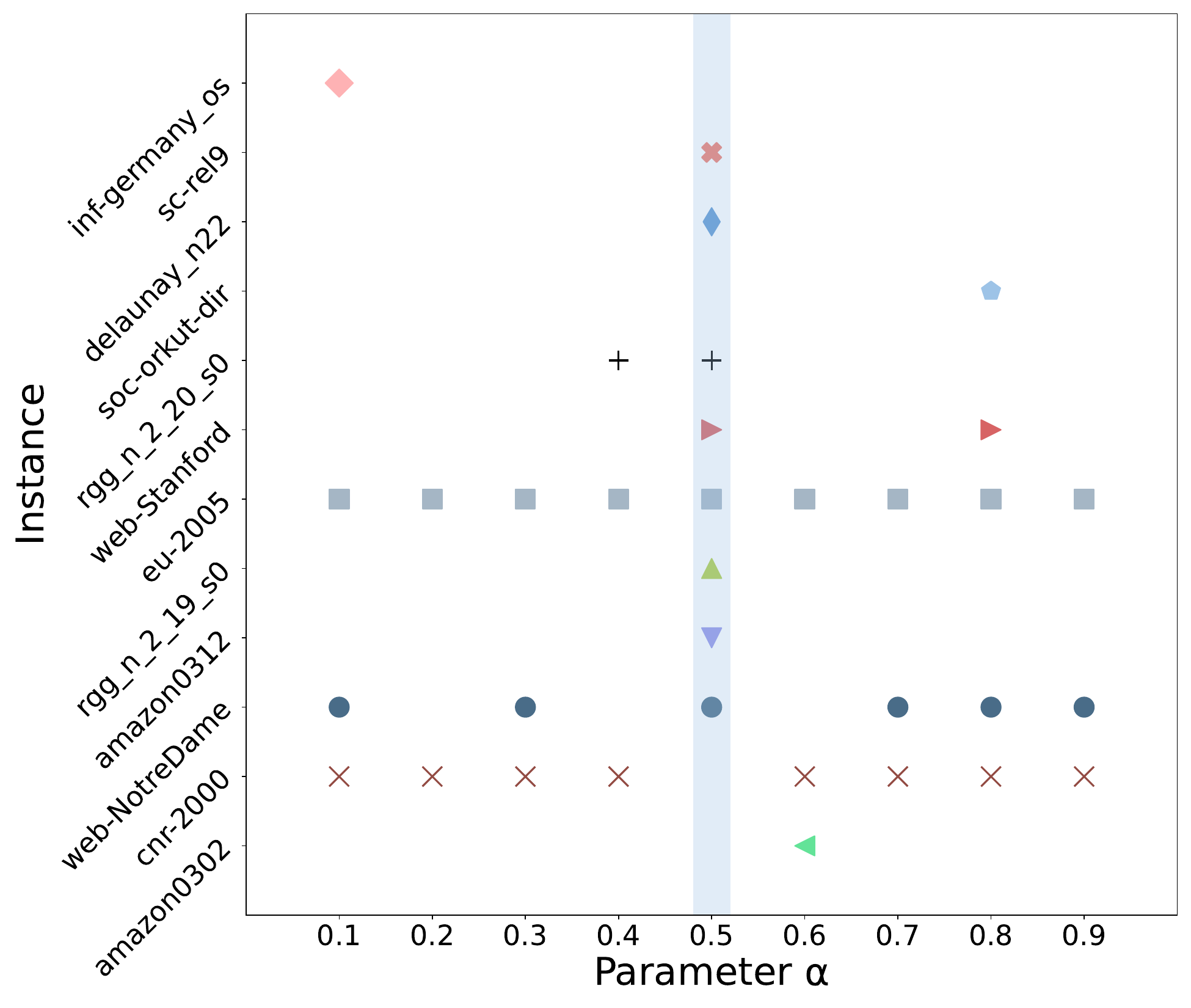}
		\caption{Results of parameter tuning for parameter for $\alpha$ in DmDS}
		\label{fig:para}
	\end{figure}
 
	\subsubsection{Parameter tuning} \label{empirical}
	We conduct an  experiment to tune the parameter $\alpha$  in DmDS (Algorithm \ref{alg3}).  We select four instances from each of SNAP, DIMACS10, and Network Repository benchmarks as representatives, in total 12 instances with vertex counts from $10^5$ to $10^7$.  We test the performance of DmDS by using distinct values  of  $\alpha$ in \{0.1, 0.2, 0.3, 0.4, 0.5, 0.6, 0.7, 0.8, 0.9\}, and the 
	experiment setup  follows that in Section \ref{setup}.  We run DmDS with a fixed value of $\alpha$ and 10 random seeds, and report the best solution.  Figure~\ref{fig:para} shows the experimental results, where the x-axis represents values of $\alpha$, and the y-axis represents instances. When DmDS finds the smallest dominating set of instance $A$ at a specific $\alpha$ value (denoted as $a$), a marker is plotted at the coordinates $(a, A)$ on the figure to indicate this finding.
	It can be seen from the data in the figure that when $\alpha=0.5$ DmDS  can find a best solution on eight (the maximum number) instances. Therefore, DmDS achieves the best performance with $\alpha=0.5$.
	
	\subsubsection{Initial Solutions}
	Table~\ref{tab:initial-1} and~\ref{tab:initial-2}  respectively report the size of initial solutions of  large sparse real-world benchmark instances obtained by  Greed() and Perturbation(), and the running time constructing the initial solutions. Note that only the instances for  which the two  solutions obtained by the two algorithms have different sizes are reported, in total 200 instances. From the data in the tables we can see that  the Greed() and the Perturbation() have a similar performance,  and each can obtain a better solution than the other on about half of the instances.  In addition,  the running time of these two algorithms is very close. This provides us with a method of improving the quality of the initial solution, that is,  constructing two solutions by  the two algorithms and selecting  a better one as the initial solution.

	\begin{table}[htbp]
		\renewcommand{\thetable}{IX-I}
		\centering
		\caption{Comparison of initial solutions constructed by two approaches}
		\resizebox{\linewidth}{!}{
			\begin{tabular}{p{6em}rrrr|p{7em}rrrr}
				\toprule
				Instance   & \multicolumn{2}{c}{Perturbation} & \multicolumn{2}{c|}{Greed} & Instance   & \multicolumn{2}{c}{Perturbation} & \multicolumn{2}{c}{Greed} \\
				Name       & \multicolumn{1}{l}{Solution} & \multicolumn{1}{l}{Time} & \multicolumn{1}{l}{Solution} & \multicolumn{1}{l|}{Time} & Name       & \multicolumn{1}{l}{Solution} & \multicolumn{1}{l}{Time} & \multicolumn{1}{l}{Solution} & \multicolumn{1}{l}{Time} \\
				\midrule
				amazon0302 & \textbf{38124} & 0.200      & 38130      & 0.185      & ca-AstroPh & 2121       & 0.021      & \textbf{2117} & 0.020  \\
				amazon0312 & \textbf{47885} & 0.505      & 47910      & 0.481      & ca-coau*dblp & 37739      & 1.622      & \textbf{37726} & 1.601  \\
				amazon0505 & \textbf{49692} & 0.521      & 49718      & 0.506      & ca-CondMat & \textbf{3032} & 0.017      & 3035       & 0.015  \\
				amazon0601 & \textbf{44824} & 0.519      & 44856      & 0.504      & ca-dblp-2012 & \textbf{46277} & 0.232      & 46279      & 0.205  \\
				belgium\_osm & 503400     & 0.743      & \textbf{501650} & 0.673      & cage15     & 503412     & 30.221     & \textbf{503381} & 25.954  \\
				bio-CE-CX  & 2681       & 0.033      & \textbf{2669} & 0.030      & ca*wood-2009 & \textbf{51409} & 9.349      & 51437      & 9.035  \\
				bio-CE-GN  & 204        & 0.006      & \textbf{202} & 0.005      & caidaRouterLevel & \textbf{40907} & 0.142      & 40929      & 0.130  \\
				bio-DM-CX  & \textbf{536} & 0.009      & 538        & 0.010      & ca-IMDB    & \textbf{120791} & 1.453      & 120826     & 1.382  \\
				bio-dmela  & \textbf{1453} & 0.006      & 1454       & 0.005      & ca-MathSciNet & 65659      & 0.216      & \textbf{65645} & 0.191  \\
				bio-DR-CX  & 341        & 0.011      & \textbf{338} & 0.010      & channel*b050 & \textbf{434264} & 7.336      & 434313     & 7.266  \\
				bio*sion-yeast & \textbf{281} & 0.002      & 282        & 0.002      & citat*seer & \textbf{43971} & 0.311      & 43984      & 0.293  \\
				bio-gr*fruitfly & \textbf{1522} & 0.006      & 1524       & 0.005      & cit-HepPh  & \textbf{960} & 0.271      & 965        & 0.269  \\
				bio-gr*human & 1792       & 0.007      & \textbf{1789} & 0.006      & cit-HepTh  & \textbf{1061} & 0.212      & 1063       & 0.209  \\
				bio-gr*yeast & \textbf{299} & 0.021      & 304        & 0.018      & cit-Patents & \textbf{641775} & 15.070     & 641946     & 16.107  \\
				bio-HS-HT  & \textbf{461} & 0.003      & 462        & 0.002      & cnr-2000   & \textbf{22079} & 0.280      & 22087      & 0.259  \\
				bio-HS-LC  & \textbf{384} & 0.006      & 385        & 0.005      & coAuthor*seer & 33336      & 0.151      & \textbf{33327} & 0.134  \\
				bio-hu*gene1 & \textbf{879} & 1.038      & 880        & 1.026      & coAuthorsDBLP & 44124      & 0.224      & \textbf{44120} & 0.204  \\
				bio-hu*gene2 & 419        & 0.725      & \textbf{418} & 0.739      & cond-mat-2005 & 5722       & 0.024      & \textbf{5718} & 0.022  \\
				bio-mou*gene & \textbf{2837} & 1.746      & 2847       & 1.484      & coPapersCiteseer & 27359      & 1.466      & \textbf{27345} & 1.463  \\
				bio-Worm*v3 & \textbf{2183} & 0.091      & 2184       & 0.089      & coPapersDBLP & \textbf{37743} & 1.632      & 37746      & 1.611  \\
				bn*864*\_1-bg & 26387      & 19.222     & \textbf{26337} & 19.254     & dbpedia-link & \textbf{1968854} & 95.694     & 1968864    & 92.473  \\
				bn*865*\_1-bg & 24410      & 22.290     & \textbf{24386} & 22.455     & delaunay\_n22 & \textbf{751637} & 4.741      & 751665     & 4.507  \\
				bn*867*\_1-bg & 28573      & 20.013     & \textbf{28564} & 20.124     & delaunay\_n23 & 1502976    & 10.717     & \textbf{1502944} & 10.292  \\
				bn*868*\_1-bg & \textbf{27543} & 20.935     & 27558      & 20.854     & delaunay\_n24 & \textbf{3004962} & 28.032     & 3005258    & 25.657  \\
				bn*868*\_2-bg & \textbf{28078} & 22.322     & 28089      & 21.883     & delicious-ti & 140725     & 12.511     & \textbf{140723} & 12.323  \\
				bn*869*\_1-bg & \textbf{26345} & 17.822     & 26347      & 17.570     & eu-2005    & 32415      & 1.672      & \textbf{32413} & 1.682  \\
				bn*869*\_2-bg & \textbf{28227} & 20.039     & 28260      & 19.826     & Freescale1 & 608329     & 2.782      & \textbf{608162} & 2.581  \\
				bn*870*\_1-bg & \textbf{31272} & 20.766     & 31276      & 20.560     & friendster & 656689     & 36.958     & \textbf{656688} & 35.978  \\
				bn*870*\_2-bg & \textbf{30652} & 23.526     & 30690      & 23.999     & in-2004    & 77997      & 1.475      & \textbf{77990} & 1.350  \\
				bn*871*\_2-bg & \textbf{25965} & 20.847     & 25970      & 21.851     & indochina-2004 & 401559     & 15.083     & \textbf{401524} & 14.284  \\
				bn*873*\_1-bg & 22094      & 17.030     & \textbf{22090} & 17.979     & inf-europe\_os & 17810003   & 47.382     & \textbf{17757076} & 41.144  \\
				bn*874*\_2-bg & 28677      & 20.769     & \textbf{28675} & 21.882     & inf-germany\_os & 4054144    & 8.876      & \textbf{4035903} & 7.616  \\
				bn*876*\_1-bg & \textbf{33023} & 17.369     & 33033      & 17.889     & inf-road\_central & \textbf{4754295} & 16.641     & 4754615    & 15.449  \\
				bn*876*\_2-bg & 33025      & 16.994     & \textbf{33000} & 17.808     & inf-roadNet-CA & \textbf{619587} & 1.117      & 620043     & 0.933  \\
				bn*878*\_1-bg & 26378      & 15.334     & \textbf{26359} & 16.307     & inf-roadNet-PA & 344442     & 0.712      & \textbf{344398} & 0.534  \\
				bn*886*\_1 & \textbf{27749} & 19.993     & 27799      & 20.884     & inf-road-usa & 8093641    & 26.477     & \textbf{8081092} & 22.328  \\
				bn*889*\_1 & \textbf{26243} & 17.245     & 26270      & 18.275     & kron\_*logn16 & \textbf{3885} & 0.329      & 3887       & 0.315  \\
				bn*890*\_2 & \textbf{26721} & 19.507     & 26795      & 20.023     & kron\_*logn17 & \textbf{7470} & 0.797      & 7471       & 0.763  \\
				bn*896*\_2-bg & 24187      & 12.883     & \textbf{24154} & 13.718     & kron\_*logn18 & \textbf{14269} & 1.859      & 14271      & 1.747  \\
				bn*911*\_2 & 28427      & 17.791     & \textbf{28409} & 19.042     & kron\_*logn19 & \textbf{27751} & 4.164      & 27752      & 3.926  \\
				bn*912*\_2 & 28096      & 17.443     & \textbf{28062} & 18.610     & kron\_*logn21 & \textbf{102246} & 25.065     & 102248     & 26.030  \\
				bn*913*\_2 & \textbf{25437} & 22.590     & 25458      & 23.534     & ljou*-2008 & 1008173    & 33.097     & \textbf{1007911} & 31.118  \\
				bn*914*\_1-bg & \textbf{28251} & 13.629     & 28259      & 14.282     & luxem*\_os & 40169      & 0.028      & \textbf{40137} & 0.023  \\
				bn*914*\_2 & \textbf{27510} & 16.953     & 27521      & 17.632     & memchip    & \textbf{479145} & 2.084      & 479158     & 1.950  \\
				bn*916*\_1 & 26078      & 20.488     & \textbf{26065} & 17.291     & p2p-*lla04 & 2249       & 0.006      & \textbf{2247} & 0.005  \\
				bn*917*\_1 & \textbf{24903} & 18.050     & 24947      & 17.998     & p2p-*lla25 & 4524       & 0.009      & \textbf{4523} & 0.008  \\
				bn*918*\_1 & \textbf{27466} & 20.545     & 27535      & 20.855     & p2p-*lla30 & 7171       & 0.016      & \textbf{7170} & 0.014  \\
				bn*M87101698 & \textbf{23988} & 25.383     & 24027      & 25.335     & patents    & \textbf{653953} & 11.875     & 653967     & 11.712  \\
				bn*M87101705 & 21453      & 53.213     & \textbf{21448} & 49.726     & rec-epinion & 9057       & 3.524      & \textbf{9055} & 3.366  \\
				bn*M87102230 & 33268      & 12.465     & \textbf{33261} & 10.875     & rec-libimseti-dir & 13000      & 3.075      & \textbf{12992} & 2.968  \\
				\bottomrule
			\end{tabular}%
		}
		\label{tab:initial-1}%
	\end{table}%
	
	\begin{table}[htbp]
		\renewcommand{\thetable}{IX-II}
		\centering
		\caption{Comparison of initial solutions constructed by two approaches}
		\resizebox{\linewidth}{!}{
			\begin{tabular}{p{5em}rrrr|p{7em}rrrr}
				\toprule
				Instance   & \multicolumn{2}{c}{Perturbation} & \multicolumn{2}{c|}{Greed} & Instance   & \multicolumn{2}{c}{Perturbation} & \multicolumn{2}{c}{Greed} \\
				Name       & \multicolumn{1}{l}{Solution} & \multicolumn{1}{l}{Time} & \multicolumn{1}{l}{Solution} & \multicolumn{1}{l|}{Time} & Name       & \multicolumn{1}{l}{Solution} & \multicolumn{1}{l}{Time} & \multicolumn{1}{l}{Solution} & \multicolumn{1}{l}{Time} \\
				\midrule
				bn*M87102575 & \textbf{21114} & 13.882     & 21139      & 13.740     & rgg\_n\_2\_15\_s0 & 4189       & 0.018      & \textbf{4186} & 0.016  \\
				bn*M87103674 & 33791      & 10.547     & \textbf{33718} & 10.878     & rgg\_n\_2\_16\_s0 & \textbf{7888} & 0.038      & 7903       & 0.034  \\
				bn*M87104201 & 18172      & 36.854     & \textbf{18125} & 37.286     & rgg\_n\_2\_17\_s0 & 15000      & 0.087      & \textbf{14989} & 0.076  \\
				bn*M87104300 & 28896      & 11.016     & \textbf{28879} & 10.971     & rgg\_n\_2\_18\_s0 & 28521      & 0.232      & \textbf{28504} & 0.182  \\
				bn*M87104509 & \textbf{31413} & 9.289      & 31440      & 8.679      & rgg\_n\_2\_19\_s0 & 54466      & 0.437      & \textbf{54424} & 0.427  \\
				bn*M87105966 & \textbf{24661} & 29.072     & 24670      & 29.768     & rgg\_n\_2\_20\_s0 & \textbf{104388} & 1.276      & 104399     & 1.255  \\
				bn*M87107242 & 36153      & 9.212      & \textbf{36118} & 8.620      & rgg\_n\_2\_21\_s0 & \textbf{199980} & 2.947      & 200053     & 2.730  \\
				bn*M87108808 & \textbf{25510} & 35.519     & 25525      & 35.411     & rgg\_n\_2\_22\_s0 & 385079     & 6.192      & \textbf{384882} & 5.759  \\
				bn*M87109786 & 33794      & 9.510      & \textbf{33772} & 10.024     & rgg\_n\_2\_23\_s0 & \textbf{741341} & 14.116     & 741546     & 13.555  \\
				bn*M87110148 & 23523      & 30.698     & \textbf{23499} & 28.078     & rgg\_n\_2\_24\_s0 & 1430548    & 31.453     & \textbf{1430305} & 30.846  \\
				bn*M87110650 & 30108      & 6.786      & \textbf{30097} & 6.876      & road\_central & \textbf{4754295} & 18.446     & 4754615    & 16.238  \\
				bn*M87110670 & \textbf{30011} & 8.022      & 30012      & 7.893      & sc-pwtk    & \textbf{4704} & 0.340      & 4709       & 0.320  \\
				bn*M87111392 & \textbf{26392} & 38.040     & 26402      & 37.394     & sc-rel9    & 143570     & 26.014     & \textbf{143244} & 25.816  \\
				bn*M87113679 & 33590      & 11.447     & \textbf{33551} & 10.654     & sc-shipsec1 & 9267       & 0.178      & \textbf{9261} & 0.168  \\
				bn*M87113878 & \textbf{16467} & 62.757     & 16470      & 60.195     & sc-shipsec5 & \textbf{12513} & 0.236      & 12557      & 0.226  \\
				bn*M87115663 & 21515      & 38.623     & \textbf{21461} & 38.154     & soc-academia & 28610      & 0.234      & \textbf{28592} & 0.217  \\
				bn*M87115834 & \textbf{33945} & 10.227     & 33961      & 9.689      & soc-buzznet & \textbf{132} & 0.386      & 133        & 0.372  \\
				bn*M87116517 & \textbf{22120} & 27.513     & 22125      & 27.647     & soc-delicious & \textbf{55758} & 0.366      & 55759      & 0.324  \\
				bn*M87116523 & \textbf{7504} & 14.141     & 7512       & 14.147     & soc-digg   & \textbf{66174} & 1.337      & 66178      & 1.270  \\
				bn*M87117093 & \textbf{26004} & 23.094     & 26014      & 23.608     & soc-dogster & 26404      & 1.551      & \textbf{26396} & 1.472  \\
				bn*M87117515 & \textbf{33690} & 8.270      & 33695      & 8.378      & socfb-A-anon & \textbf{201827} & 12.328     & 201832     & 11.944  \\
				bn*M87118219 & \textbf{26292} & 23.245     & 26316      & 23.599     & socfb-B-anon & 187069     & 10.239     & \textbf{187064} & 9.886  \\
				bn*M87118347 & \textbf{9048} & 13.867     & 9052       & 15.092     & soc-flickr & \textbf{98092} & 0.656      & 98102      & 0.598  \\
				bn*M87118759 & \textbf{32969} & 10.436     & 32988      & 10.620     & soc-flickr-und & \textbf{295751} & 4.091      & 295758     & 3.868  \\
				bn*M87118954 & 22917      & 12.496     & \textbf{22895} & 12.430     & soc-FourSquare & 61015      & 0.878      & \textbf{61013} & 0.804  \\
				bn*M87119044 & 7891       & 12.232     & \textbf{7877} & 11.666     & soc-livejournal & \textbf{796756} & 16.245     & 796910     & 15.828  \\
				bn*M87119472 & 31168      & 11.373     & \textbf{31161} & 11.640     & soc-live*groups & \textbf{1071441} & 53.229     & 1071446    & 52.440  \\
				bn*M87121714 & \textbf{32745} & 8.413      & 32815      & 8.203      & soc-LiveMocha & 1470       & 0.316      & \textbf{1468} & 0.294  \\
				bn*M87122310 & \textbf{19852} & 14.124     & 19860      & 14.383     & soc-ljou*-2008 & 1008211    & 28.282     & \textbf{1007895} & 27.316  \\
				bn*M87123142 & 34887      & 9.050      & \textbf{34880} & 8.972      & soc-orkut  & \textbf{119979} & 53.192     & 120062     & 52.525  \\
				bn*M87123456 & 38037      & 9.129      & \textbf{37970} & 9.283      & soc-orkut-dir & 100639     & 72.095     & \textbf{100557} & 71.883  \\
				bn*M87124029 & \textbf{31005} & 7.393      & 31010      & 7.073      & soc-pokec  & \textbf{213224} & 8.133      & 213324     & 8.691  \\
				bn*M87124152 & \textbf{19542} & 30.825     & 19611      & 30.198     & Slashdot0902 & 15310      & 0.082      & \textbf{15309} & 0.076  \\
				bn*M87124563 & 19306      & 31.539     & \textbf{19249} & 33.387     & soc-twitter-higgs & \textbf{15054} & 2.547      & 15066      & 2.586  \\
				bn*M87124670 & \textbf{29437} & 7.396      & 29439      & 6.644      & soc-youtube & \textbf{89767} & 0.531      & 89772      & 0.482  \\
				bn*M87125286 & \textbf{22401} & 42.058     & 22407      & 42.424     & soc-you*snap & \textbf{213131} & 1.129      & 213135     & 1.015  \\
				bn*M87125330 & \textbf{34915} & 9.187      & 35001      & 8.754      & tech-as-skitter & \textbf{183884} & 2.916      & 183889     & 2.697  \\
				bn*M87125334 & \textbf{34962} & 6.482      & 35064      & 6.111      & tech-ip    & \textbf{176} & 5.244      & 177        & 5.096  \\
				bn*M87125521 & 34982      & 7.748      & \textbf{34943} & 7.169      & uk-2002    & \textbf{1043790} & 35.016     & 1043969    & 33.128  \\
				bn*M87125691 & 36231      & 8.042      & \textbf{36230} & 7.465      & wb-edu     & 836877     & 9.882      & \textbf{836640} & 8.624  \\
				bn*M87125989 & 37310      & 8.943      & \textbf{37298} & 8.840      & web-arabic-2005 & 17007      & 0.135      & \textbf{16998} & 0.119  \\
				bn*M87126525 & 18473      & 25.081     & \textbf{18461} & 25.324     & web-baidu-baike & \textbf{277019} & 9.583      & 277024     & 9.187  \\
				bn*M87127186 & \textbf{20255} & 35.611     & 20292      & 36.168     & web-BerkStan & \textbf{29523} & 0.695      & 29541      & 0.655  \\
				bn*M87127667 & \textbf{22967} & 12.790     & 22984      & 12.771     & web-Google & \textbf{80047} & 1.084      & 80059      & 1.008  \\
				bn*M87127677 & 26478      & 45.005     & \textbf{26448} & 42.335     & web-it-2004 & 33003      & 0.503      & \textbf{33001} & 0.477  \\
				bn*M87128194 & 26530      & 25.114     & \textbf{26521} & 25.687     & web-NotreDame & 23777      & 0.151      & \textbf{23774} & 0.132  \\
				bn*M87128519 & 24381      & 34.280     & \textbf{24340} & 33.890     & web-Stanford & 13968      & 0.349      & \textbf{13953} & 0.328  \\
				bn*M87129974 & 36857      & 9.099      & \textbf{36828} & 9.167      & web-wiki*2009 & \textbf{347849} & 2.402      & 347917     & 2.277  \\
				c-62ghs    & \textbf{15194} & 0.028      & 15196      & 0.026      & web-wiki*rowth & \textbf{117807} & 15.423     & 117809     & 14.910  \\
				c-66b      & \textbf{21193} & 0.026      & 21197      & 0.024      & wikipedia\_link\_en & 212878     & 14.861     & \textbf{212876} & 14.428  \\
				\bottomrule
			\end{tabular}%
		}
		\label{tab:initial-2}%
	\end{table}%
	
	\subsubsection{Discussion}
	According to the performance of the three algorithms on different instances, some observations are highlighted as follows. First, DmDS can find a series of solutions with different sizes within 1,000 seconds,  due to the intrinsic mechanism of the local search. However, 
	a single search strategy may lead to the algorithm becoming trapped in a local optimum, which may reduce the efficiency of solution improvement. DmDS   integrates two distinct search strategies (i.e.,  (2,1)-swaps and (3,2)-swaps) and conducts  tiny perturbations in the process of searching solutions, by which the diversity of solutions is enhanced  and the issue of  local optima is also weakened. 
	Next, a strict tabu strategy restricts the search space, which also reduces the rate of improving a solution.  DmDS   does not use tabu strategies, for which all vertices have an opportunity to be searched. The above analysis may be an explanation that DmDS has a good performance on the testing benchmark instances.

	\section{Conclusion} \label{sec6}
	In this paper, we propose an efficient local search algorithm for the MinDS problem named DmDS. Specifically, DmDS introduces a dual-mode local search framework in the search phase, a novel approach based on greedy strategy and perturbation mechanism to improve the quality of the initial solution, and a new criterion for vertex selection.
	The experimental results show that DmDS significantly outperforms other state-of-the-art MinDS algorithms on seven benchmarks including  346 instances (or families). In the future, we would like to apply the three proposed  techniques to solve other  intractable graph problems, such as the minimum vertex cover.


	
	\bibliographystyle{IEEEtran}
	\bibliography{ref}

\end{document}